\documentclass[a4paper]{article}

\title{An encoding of array verification problems into array-free Horn clauses\thanks{The research leading to these results has received funding from the \href{http://erc.europa.eu/}{European Research Council} under the European Union's Seventh Framework Programme (FP/2007-2013) / ERC Grant Agreement nr.~306595 \href{http://stator.imag.fr/}{\mbox{``STATOR''}}.}}

\author{David Monniaux\\
{\small Univ. Grenoble Alpes, VERIMAG, F-38000 Grenoble}\\
{\small CNRS, VERIMAG, F-38000 Grenoble, France}
\and Laure Gonnord\\
{\small LIP, Univ. Lyon-1, France}
}

\usepackage{paralist}

\usepackage{amsmath,amsfonts,amsthm,galois,proof}

\usepackage{tikz}
\usetikzlibrary{arrows,automata,shapes}

\usepackage{listings}
\lstdefinelanguage{imp}[]{C}{morekeywords={assume,assert}}
\lstdefinelanguage{SMT}[]{Lisp}{}
\usepackage[natbib=true,backend=bibtex]{biblatex}
\bibliography{arrays_Horn}

\usepackage[pdftitle={An encoding of array verification problems into array-free Horn clauses},pdfauthor={David Monniaux and Laure Gonnord}]{hyperref}

\newcommand{\ve}[1]{\mathbf{#1}}
\newcommand{\vx}{\ve{x}}

\newcommand{\ZZ}{\mathbb{Z}}
\newcommand{\QQ}{\mathbb{Q}}
\newcommand{\NN}{\mathbb{N}}
\newcommand{\arraytype}[2]{\mathit{Array}\left(#1,#2\right)}

\newcommand{\abstr}[1]{#1^\sharp}
\newcommand{\parts}[1]{\mathcal{P}\left(#1\right)}
\newcommand{\abstraction}[2][]{\alpha_{#1}\left(#2\right)}
\newcommand{\concretization}[2][]{\gamma_{#1}\left(#2\right)}
\newcommand{\hash}{\#}
\DeclareMathOperator{\card}{card}

\newcommand{\soft}[1]{\textsc{#1}}

\theoremstyle{definition}
\newtheorem{definition}{Definition}
\newtheorem{algo}{Algorithm}
\theoremstyle{plain}
\newtheorem{theorem}{Theorem}
\newtheorem{lemma}[theorem]{Lemma}
\newtheorem{example}{Example}
\newtheorem*{example*}{Example}
\newtheorem{remark}{Remark}

\newcommand{\rulespacing}{\\[0.4em]}
\allowdisplaybreaks

\begin{document}

\maketitle

\begin{abstract}
Automatically verifying safety properties of programs is hard, and it is even harder if the program acts upon arrays or other forms of maps.
Many approaches exist for verifying programs operating upon Boolean and integer values (e.g. abstract interpretation, counterexample-guided abstraction refinement using interpolants), but transposing them to array properties has been fraught with difficulties.

In contrast to most preceding approaches, we do not introduce a new abstract domain or a new interpolation procedure for arrays. Instead, we generate an abstraction as a scalar problem and feed it to a preexisting solver, with tunable precision.

Our transformed problem is expressed using Horn clauses, a common format with clear and unambiguous logical semantics for verification problems.
An important characteristic of our encoding is that it creates a nonlinear Horn problem, with tree unfoldings, even though following ``flatly'' the control-graph structure ordinarily yields a linear Horn problem, with linear unfoldings.
That is, our encoding cannot be expressed by an encoding into another control-flow graph problem, and truly leverages the capacity of the Horn clause format.

We illustrate our approach with a completely automated proof of the functional correctness of selection sort.
\end{abstract}

\section{Introduction}
Formal program verification, that is, proving that a given program behaves correctly according to specification in all circumstances, is difficult.
Except for very restricted classes of programs and properties, it is an undecidable question.
Yet, a variety of approaches have been developed over the last 40 years for automated or semi-automated verification, some of which have had industrial impact.

In this article, we consider programs operating over arrays, or, more generally, \emph{maps} from an index type to a value type. (in the following, we shall use ``array'' and ``map'' interchangeably).
Such programs contain read (e.g. $v:=a[i]$) and write ($a[i]:=v$) operations over arrays, as well as ``scalar'' operations.%
\footnote{In the following, we shall lump as ``scalar'' operations all operations not involving the array under consideration, e.g. $i:=i+1$. Any data types (integers, strings etc.) are supported provided that they are supported by the back-end solver.}

\paragraph{Universally quantified properties}
Very often, desirable properties over arrays are universally quantified; e.g. sortedness may be expressed as
$\forall k_1,k_2~ k_1 < k_2 \implies a[k_1] \leq a[k_2]$.
However, formulas with universal quantification and linear arithmetic over integers and at least one predicate symbol (a predicate being a function to the Booleans) are so expressive that one can define the execution of a Turing machine as a model to such a formula, whence this class is undecidable \cite{Halpern91}.
Some decidable subclasses have however been identified \cite{BradleyMS06}.

There is therefore no general algorithm for checking that such invariants hold, let alone inferring them. Yet, there have been several approaches proposed to infer such invariants (more on this in Section~\ref{sec:related}).
In this article, we propose a method for inferring such universally quantified invariants, given a specification on the output of the program.
Because of undecidability, this approach may fail to terminate in the general case.

Our approach is based on conversion to Horn clauses, a popular format for program verification problems~\cite{DBLP:conf/vstte/RummerHK13} supported by a number of tools.
Most conversions to Horn clauses map variables and operations from the program to variables of the same type and the same operations in the Horn clause problem:%
\footnote{With the exception of pointers and references, which need special handling and may be internally converted to array accesses.}
an integer is mapped to an integer, an array to an array, etc. If some data types are not supported by the back-end analysis, the variables of these types may be discarded, at the expense of precision --- thus if the back-end analysis does not support arrays, array reads are abstracted as nondeterministic choices, array writes are discarded, and scalar operations are mapped ``as is''.
In contrast, our approach abstracts programs much less violently, with tunable precision, even though the result still is a Horn clause problems without arrays.
Section~\ref{sec:abstraction1} explains how many properties (e.g. initialization) can be proved using one ``distinguished cell'', Section~\ref{sec:sortedness} explains how properties such as sortedness can be proved using two of them; completely discarding arrays corresponds to using zero of them.

We illustrate this approach with an automated proof that the output of \emph{selection sort} is sorted: we apply Section~\ref{sec:sortedness} to obtain a system of Horn clauses without arrays, which we feed to the \soft{Spacer} solver, which produces a model of this system, meaning that the sortedness postcondition truly holds.
Note that \soft{Spacer} cannot, on its own, reason about universal properties on arrays.

Previous approaches \cite{Monniaux_Alberti_SAS2015} using ``distinguished cells'' amounted (even though they were not described as such) to linear Horn rules; on contrast, our abstract semantics uses non-linear Horn rules, which leads to higher precision (Sec.~\ref{sec:abstraction_weak_transfo}).

\paragraph{Multiset of contents}
It is often necessary to reason not only about individual elements of an array or map, but also about its contents as a whole: e.g. sorting algorithms preserve the contents of the array (even though, locally, when moving elements around, they may break this invariant).

The multiset of the contents of an array of elements of type $\beta$ is a map from $\beta$ to $\NN$.
Using that remark, we can abstract the array both using our ``distinguish cell'' approach and as the multiset of its elements (Sec.~\ref{sec:multiset_of_elements}); we provide suitable program transformations.

We illustrate that approach with an automated proof that the output of selection sort has the same contents as its input (that is, the output is a permutation of the input).

\paragraph{Contributions}
Our main contribution is a system of rules for transforming the atomic program statements in a program operating over arrays or maps, as well as (optionally) the universally quantified postcondition to prove, into a system of non-linear Horn clauses over scalar variables only, with tunable precision.
Statements operating over non-arrays variables are mapped (almost) identically to their concrete semantics.
This system over-approximates the behavior of the program.
A solution of that system can be mapped to inductive invariants over the original programs, including universally properties over arrays.

A second contribution, based upon the first, is a system of rules of the same kind that also keeps tracks of array/map contents.

We illustrate both these systems with automated proofs of functional correctness of array initialization, array reversal and selection sort.
For each of these proofs, we simply apply our transformation rules and apply a third-party solver for Horn clauses over scalars.
We also show the user can optionally help the solver converge faster by supplying partial invariants.

A third contribution is a counterexample reconstruction scheme (Sec.~\ref{sec:counterexample}), if the property to prove is actually false.

\section{Program Verification as solving Horn clauses}
A classical approach to program analysis is to consider a program as a control-flow graph and to attach to each vertex $p_i$ (control point) an \emph{inductive invariant} $I_i$: a set of possible values $\vx$ of the program variables (and memory stack and heap, as needed) so that
\begin{inparaenum}[i)]
\item
the set associated to the initial control point $p_{i_0}$ contains the possible initialization values $S_{i_0}$
\item
for each edge $p_i \rightarrow_c p_j$, the set $I_j$ associated to the target control point $p_j$ should include all the states reachable from the states in the set $I_i$ associated to the source control point $p_i$ according to the transition relation $\tau_{i,j}$ of the edge.
\end{inparaenum}
Inductiveness is thus defined by \emph{Horn clauses}:
\begin{eqnarray}
\forall \vx,~ \vx \in S_{i_0} \implies \vx \in I_{i_0} \label{clause:init} \\
\forall \vx,\vx'~ \vx \in I_i \land (\vx,\vx') \in \tau_{i,j}
  \implies \vx' \in I_j \label{clause:inductive}
\end{eqnarray}

For proving safety properties, in addition to inductiveness, one requires that error locations $p_{e_1},\dots,p_{e_n}$ are proved to be unreachable (the associated set of states is empty): this amounts to Horn clauses implying false ($\bot$):
$\forall \vx,~ \vx \in I_{e_i} \implies \bot$.

Various tools can solve such systems of Horn clauses, that is, can synthesize suitable predicates $I_i$, which constitute \emph{inductive invariants}.
In this article, we tried \soft{Z3}%
\footnote{\url{https://github.com/Z3Prover}}
with the \soft{PDR} fixed point solver \cite{DBLP:conf/sat/HoderB12},
Z3 with the \soft{Spacer} solver \cite{DBLP:conf/cav/KomuravelliGCC13,DBLP:conf/cav/KomuravelliGC14},%
\footnote{\url{https://bitbucket.org/spacer/code}}
and \soft{Eldarica}%
\footnote{\url{https://github.com/uuverifiers/eldarica}}
\cite{DBLP:conf/cav/RummerHK13}.%
\footnote{This list is not exhaustive; we apologize to authors of other tools.
Neither did we conduct systematic comparisons between these three tools, which each have numerous configuration options:
for our purposes, it sufficed that at least one concluded within reasonable time that our property was proved.}
Since program verification is undecidable, such tools, in general, may fail to terminate, or may return ``unknown''.

For the sake of simplicity, we shall consider, in this article, that all integer variables in programs are mathematical integers ($\ZZ$) as opposed to machine integers.%
\footnote{%
A classical approach is to add overflow checks to the intermediate representation of programs in order to be able to express their semantics with mathematical integers even though they operate over machine integers.}
In our semantics, we consider that reads from out-of-range locations (array index out of bounds, buffer overflows) stop the execution of the program immediately, but that a write to a location makes this location defined.
Again, it is easy to modify our semantics to include systematic array bound checks, jumps to error conditions, etc.

In examples, instead of writing $I_{\mathit{stmt}}$ for the name of the predicate (inductive invariant) at statement $\mathit{stmt}$, we shall write $\mathit{stmt}$ directly, for readability's sake: thus we write e.g. $\mathit{loop}$ for a predicate at the head of a loop.

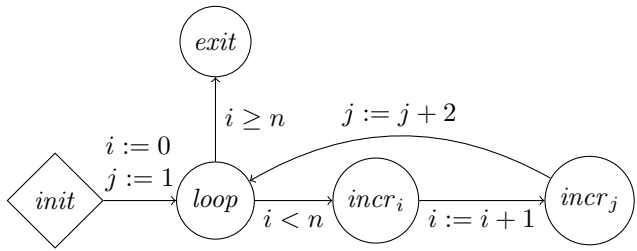
\begin{figure}
\begin{center}
\begin{tikzpicture}[->, node distance=6em]
\node[state, initial by diamond] (init) {$\mathit{init}$};
\node[state] (loop) [right of=init] {$\mathit{loop}$};
\node[state] (incri) [right of=loop] {$\mathit{incr}_i$};
\node[state] (incrj) [node distance=8em, right of=incri] {$\mathit{incr}_j$};
\node[state] (exit) [above of=loop, node distance=6em] {$\mathit{exit}$};
\path (init) edge node[above, align=left] {$i := 0$\\$j := 1$} (loop);
\path (loop) edge node[below] {$i < n$} (incri);
\path (incri) edge node[below] {$i := i+1$} (incrj);
\path (incrj) edge[bend right] node[above] {$j := j+2$} (loop);
\path (loop) edge node[right] {$i \geq n$} (exit);
\end{tikzpicture}
\end{center}

\caption{The control-flow graph for Program~\ref{prog:loop_ij}}
\label{fig:loop_ij}
\end{figure}

\begin{example}\label{ex:loop_ij}
Consider the following program
\begin{lstlisting}[caption={A simple loop without arrays},label=prog:loop_ij]
int loop_ij(int n) {
  int i = 0, j = 1;
  while (i < n) {
    i = i + 1;
    j = j + 2;
  }
}
\end{lstlisting}

Its semantics gets encoded into Horn \emph{rules} as predicates over triples $(n,i,j)$, one predicate for each node of the control-flow graph in Figure~\ref{fig:loop_ij}:
\begin{align}
\forall n \in \ZZ~ \mathit{loop}(n, 0, 1)\\
\forall n,i,j \in \ZZ~
  \mathit{loop}(n, i, j) \land i < n \implies
  \mathit{incr}_i(n, i, j)\\
\forall n,i,j \in \ZZ~
  \mathit{loop}(n, i, j) \land i \geq n \implies
  \mathit{exit}(n, i, j)\\
\forall n,i,j \in \ZZ~
  \mathit{incr}_i(n, i, j) \implies
  \mathit{incr}_j(n, i+1, j)\\
\forall n,i,j \in \ZZ~
  \mathit{incr}_j(n, i, j) \implies
  \mathit{loop}(n, i, j+2)
\end{align}

If we wish to prove that, at the end of the program, $n\geq 0 \implies i=n$, we add the Horn \emph{query}
\begin{align}
\forall n,i,j\in \ZZ~ \mathit{exit}(n,i,j) \land n\geq 0 \implies i=n
\label{formula:query}
\end{align}
\soft{Spacer} and \soft{Z3/PDR} then answer ``satisfiable'' after \emph{synthesizing} suitable predicates $\mathit{loop}$, $\mathit{incr}_i$ etc. satisfying the Horn system --- otherwise said, inductive invariants implying the postcondition~\ref{formula:query}.

If we had made the mistake of forgetting that $i=n$ holds finally only for $n \geq 0$, we would have written the query as
\begin{align}
\forall n,i,j\in \ZZ~ \mathit{exit}(n,i,j) \implies i=n
\end{align}
and these solvers would have answered ``unsatisfiable''.%
\footnote{Horn clause solvers based on counterexample-based refinement are rather good at handling disjunctions (here, $n < 0$ vs $n \geq 0$). Tools based on convex domains, such as polyhedra, may have more difficulties.}

Let us now try proving that $n \geq 0 \Rightarrow j \leq 2+3n$ holds finally:
\begin{align}
\forall n,i,j\in \ZZ~ \mathit{exit}(n,i,j) \land n\geq 0 \implies j \leq 2+3n
\end{align}
While \soft{Spacer} answers instantaneously, \soft{Z3/PDR} seems to enter a neverending sequence of refinements.
\end{example}

For reasons of efficiency of the back-end solver, it may be desirable to have fewer predicates. It is possible to automatically simplify the system of Horn rules by coalescing several rules together: for instance, if we have $\forall x,y,~ I_1(x,y) \Rightarrow I_2(x+1,y)$ and $\forall y,z,~ I_2(x,y) \Rightarrow I_3(x,y+2)$, and $I_2$ does not occur elsewhere, then we can remove $I_2$ and use a single rule $\forall x,y,~ I_1(x,y) \implies I_3(x+1,y+2)$.
Similarly, if in the antecedent of a Horn rule we have an equality $x=e$ where $e$ is an expression and $x$ is universally quantified in the rule, then we can remove this variable and replace it with~$e$ in the rest of the rule.
We shall often apply such syntactic simplifications to make examples shorter and more readable.

The \emph{flat} encoding of the program described by initialization (Formulas~\ref{clause:init}) and inductiveness (Formulas~\ref{clause:inductive}), following the control-flow graph (CFG), results in a \emph{linear} system of Horn clauses, in the sense that if one ``unfolds'' the system by repeatedly rewriting the right-hand sides of implications $l_1 \land \dots \land l_n \implies r$ into their left-hand side in the rest of the Horn clauses, one gets a list, not a tree structure.

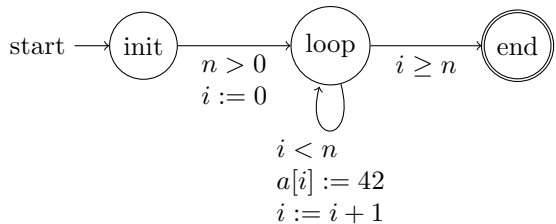
\begin{figure}
\begin{center}
\begin{tikzpicture}[->, node distance=7em]
\node[state,initial] (init) {init};
\node[state] (loop) [right of=init] {loop};
\node[state,accepting] (end) [right of=loop] {end};
\path (init) edge node[below, align=left] {$n > 0$\\$i := 0$} (loop);
\path (loop) edge [loop below] node[align=left] {
  $i < n$ \\ $a[i]:=42$ \\ $i:=i+1$} (loop);
\path (loop) edge node[below] {$i \geq n$} (end);
\end{tikzpicture}
\end{center}
\caption{Compact control-flow graph for Program~\ref{prog:array_fill1}}
\label{fig:array_fill1_cfg}
\end{figure}

\begin{example}\label{ex:array_fill1}
Consider the program:
\begin{lstlisting}[label=prog:array_fill1,caption={1D array fill}]
void array_fill1(int n, int a[n]) {
  int i = 0;
  while(i < n) {
    a[i] = 42;
    i = i+1;
  }
}
\end{lstlisting}
We would like to prove that this program truly fills array \lstinline|a[]| with value 42. The flat encoding into Horn clauses assigns a predicate (set of states) to each of the control nodes (Fig.~\ref{fig:array_fill1_cfg}), and turns each transition into a Horn rule:
\begin{align}
\begin{aligned}
\forall n\in\ZZ~ \forall a\in\arraytype{\ZZ}{\ZZ}~ n>0 \implies
  \mathit{loop}(n,0,a)
\end{aligned} \label{equ:array_fill1_arrays_begin}\\
\begin{aligned}
\forall n,i\in\ZZ~ \forall a\in \arraytype{\ZZ}{\ZZ}~
  i<n \land \mathit{loop}(n,i,a) \\
  \implies \mathit{loop}(n,i+1,\mathit{store}(a,i,42))
\end{aligned}\rulespacing
\begin{aligned}
\forall n,i\in\ZZ~ \forall a\in \arraytype{\ZZ}{\ZZ}~
  i\geq n \land \mathit{loop}(n,i,a)\\
  \implies \mathit{end}(n,a)
\end{aligned}\rulespacing
\begin{aligned}
\forall n\in\ZZ~ \forall a\in \arraytype{\ZZ}{\ZZ}~
  0 \leq x < n \land \mathit{end}(n,a) \\
  \implies a[x] = 42
\end{aligned} \label{equ:array_fill1_arrays_end}
\end{align}
where $\mathit{store}(a,i,v)$ is array $a$ where the value at index $i$ has been replaced by~$v$.
\end{example}

None of the tools we have tried (\soft{Z3/PDR}, \soft{Spacer}, \soft{Eldarica}) has been able to solve this system, presumably because they cannot infer universally quantified invariants over arrays.
Indeed, here the invariant needed in the loop is
\begin{equation}
0 \leq i \leq n \land (\forall k~ 0 \leq k < i \implies a[k] = 42)
\end{equation}
While $0 \leq i \leq n$ is inferred by a variety of approaches, the rest of the formula is a tougher problem.

Most software model checkers attempt constructing invariants from \emph{Craig interpolants} obtained from refutations of the accessibility of error states in partial unfoldings of the problem, but interpolation over array properties is difficult,
especially since the goal is not to provide any interpolant, but interpolants that generalize well to invariants \cite{Alberti_Monniaux_SAC-SVT2015,AlbertiBGRS14}.

This article instead introduces a way to derive universally quantified invariants from the analysis of a system of Horn clauses on scalar variables (without array variables).

The flat encoding is not the only possible one. One may for instance instead choose to find invariants not as sets of states (\emph{unary} predicates on states), but as \emph{binary} relations on states: a procedure or function, or in fact any part of the program with one single entry and one single exit point (e.g. a loop bosdy with no break statement) is represented by a set of input-output pairs.
In general, e.g. when a procedure encoded in this way calls itself twice in a row, the resulting system of Horn clauses is \emph{nonlinear}: unfolding the Horn clauses may lead to an exponentially growing tree~\cite{DBLP:conf/vstte/RummerHK13} (see Fig.~\ref{fig:counterexample_unfolding} for an example of a tree unfolding of a nonlinear system).
This is one reason why the Horn format for program verification is richer and more flexible than a mere CFG.
In this article, we are going to exploit nonlinear systems of Horn clauses even if encoding a CFG ``flatly''.

\section{Getting rid of the arrays}
\label{sec:abstraction1}
To use the power of Horn solver on array-free problems, we soundly abstract problems with arrays to problems without arrays.

In the Horn clauses for example~\ref{ex:array_fill1}, we attached to each program point $p_k$ a predicate $I_k$ over, say, $\ZZ \times \ZZ \times \arraytype{\ZZ}{\ZZ}$ when the program variables are two integers $i,n$ and one integer-value, integer-indexed array $a$. In any solution of the system of clauses, $\neg I_k(i,n,a)$ implies that $i,n,a$ cannot be reached at program point $p_k$. Instead, we will consider a predicate $\abstr{I}_k$ over $\ZZ \times \ZZ \times \ZZ \times \ZZ$ such that $\neg \abstr{I}_k(i,n,k,a_k)$ implies that there is at $p_k$ no reachable state $(i,n,a)$ such that $a[x]=a_k$. We thus have to provide abstract transformers for each statement.

Without loss of generality, any statement in the program can be assumed to be either
\begin{enumerate}[i)]
\item an array read to a fresh variable,
  \lstinline|v=a[i];| in C syntax, $v:=a[i]$ in pseudo-code;
  the variables of the program are $(\vx,i)$ before the statement
  and $(\vx,i,v)$ after the statement, where $\vx$ is a vector of arbitrarily
  many variables;
\item an array write, \lstinline|a[i]=v;| (where \lstinline|v| and \lstinline|i| are variables) in C syntax,  $a[i]:=v$ in pseudo-code;
  the variables of the program are $(\vx,i,v)$ before and after the statement;
\item a scalar operation, including assignments and guards over scalar variables.
\end{enumerate}
More complex statements can be transformed to a sequence of such statements, by introducing temporary variables if needed: for instance, $a[i]:=a[j]$ is transformed into $\mathit{temp}:=a[j];~ a[i]:=\mathit{temp}$.

\begin{definition}[Read statement]\label{def:read1}
Let \lstinline|v| be a variable of type $\beta$,
\lstinline|i| be a variable of type $\iota$,
and \lstinline|a| be an array of values of type $\beta$ with an index of type $\iota$.
Let $\vx$ be the other program variables, taken in~$\chi$.
The concrete ``next state'' relation for the read statement \lstinline|v=a[i];|
is $(\vx,i,a) \rightarrow_c (\vx,i,a[i],a)$.

Its forward abstract semantics is encoded into two Horn clauses, assuming the statement is between locations $p_1$ and $p_2$:
\begin{align}\label{rule:read1_different}
\begin{aligned}
\forall \vx \in \chi~ \forall i \in \iota~ \forall a_i \in \beta~
  \forall k \in \iota~ \forall a_k \in \beta\\
  k \neq i \land \abstr{I}_1\big((\vx,i),(k,a_k)\big) \land
                 \abstr{I}_1\big((\vx,i),(i,a_i)\big) \\ \implies
  \abstr{I}_2\big((\vx,a_i,i),(k,a_k)\big)
\end{aligned}\rulespacing\label{rule:read1_same}
\begin{aligned}
\forall \vx \in \chi~ \forall i \in \iota~ \forall a_i \in \beta~
  \forall k \in \iota~ \forall a_k \in \beta\\
  \abstr{I}_1\big((\vx,i),(i,a_i)\big) \implies
  \abstr{I}_2\big((\vx,a_i,i),(i,a_i)\big)
\end{aligned}
\end{align}
\end{definition}

While rule~\ref{rule:read1_same} is straightforward, the nonlinear rule~\ref{rule:read1_different} may be more difficult to comprehend. The intuition is that, to have $a_i = a[i]$ and $a_k = a[k]$ at the read instruction with a given valuation $(\vx,i)$ of the other variables, both $a_i = a[i]$ and $a_k = a[k]$ had to be reachable with the same valuation.

\begin{remark}
One weakens the semantics by replacing these two rules by a single Rule~\ref{rule:read1_different} without the $i \neq k$ guard. Rule \ref{rule:read1_same} ensures that in the outcome, if $i = k$ then $v = a_k$.
\end{remark}

\begin{definition}[Write statement]\label{def:write1}
With the same notations as above. 
The concrete ``next state'' relation for the write statement \lstinline|a[i]=v;| is
$(\vx,i,v,a) \rightarrow_c (\vx,i,v,\mathit{store}(a,i,v))$.

Its forward abstract semantics is encoded into two Horn clauses:
\begin{align}
\begin{aligned}
\forall \vx \in \chi~ \forall i \in \iota~ \forall v \in \beta~
  \forall k \in \iota~ \forall a_k \in \beta\\
  \abstr{I}_1\big((\vx,i,v),(k,a_k)\big) \land i\neq k \implies
  \abstr{I}_2\big((\vx,v,i),(k,a_k)\big)
\end{aligned}\label{rule:write1_different}\rulespacing
\begin{aligned}
\forall \vx \in \chi~ \forall i \in \iota~ \forall v \in \beta~
  \forall k \in \iota~ \forall a_k \in \beta\\
  \abstr{I}_1\big((\vx,i,v),(i,a_k)\big) \implies
  \abstr{I}_2\big((\vx,v,i),(i,v)\big)
\end{aligned}\label{rule:write1_same}
\end{align}
\end{definition}

\begin{definition}[Initialization]\label{def:initialization1}
Creating an array variable with nondeterministically chosen initial content is abstracted by
\begin{align}
\begin{aligned}
\forall \vx \in \chi~ \forall k \in \iota~ \forall a_k \in \beta~
  \abstr{I}_1(\vx) \implies
  \abstr{I}_2(\vx,k,a_k)
\end{aligned}
\end{align}

In particular, creating an array variable indexed by $0 \dots n-1$ is abstracted by:
\begin{align}
\begin{aligned}
\forall \vx \in \chi~ \forall k \in \ZZ~ \forall a_k \in \beta~
  \abstr{I}_1(\vx) \land 0 \leq k < n \\ \implies
  \abstr{I}_2(\vx,k,a_k)
\end{aligned}
\end{align}
\end{definition}

\begin{remark}
Because the $0 \leq k <n$ condition gets naturally propagated throughout the rules (it holds at the initialization state and can be assumed to hold at other states), in our examples, we shall often omit this condition from the other rules, for the sake of brevity, and simply write $k \in \ZZ$.
\end{remark}

\begin{definition}[Scalar statements]\label{def:scalar1}
With the same notations as above, we consider a statement (or sequence thereof) operating only on scalar variables: $\vx \rightarrow_s \vx'$ if it is possible to obtain scalar values $\vx'$ after executing the statement on scalar values $\vx$.
The concrete ``next state'' relation for that statement is
$(\vx,i,v,a) \rightarrow_c (\vx',i,v,a)$.

Its forward abstract semantics is encoded into one Horn clause:
\begin{align}
\begin{aligned}
\forall \vx \in \chi~ \forall k \in \iota~ \forall a_k \in \beta\\
  \abstr{I}_1(\vx,k,a_k) \land \vx \rightarrow_s \vx'
  \implies \abstr{I}_2(\vx',k,a_k)
\end{aligned}
\end{align}
\end{definition}

\begin{example}
A test $x \neq y$ gets abstracted as
\begin{equation}
\forall x,y,k,a_k~ \abstr{I}_1(x,y,k,a_k) \land x \neq y \implies
                    \abstr{I}_2(x,y,k,a_k)
\end{equation}
\end{example}

\begin{definition}\label{def:kill}
The scalar operation $\mathit{kill}(v_1,\dots,v_n)$ removes variables $v_1,\dots,v_n$: $(\vx,v_1,\dots,v_n) \rightarrow \vx$.
\end{definition}

We shall apply it to get rid of dead variables, sometimes, for the sake of brevity, without explicit note, by coalescing it with other operations.

We use the same Galois connection \cite{DBLP:journals/logcom/CousotC92} as some earlier works \cite{Monniaux_Alberti_SAS2015} \cite[Sec.~2.1]{DBLP:conf/iccl/CousotC94}:
\begin{definition}\label{def:Galois1}
The \emph{concretization} of $\abstr{I} \subseteq \chi \times (\iota \times \beta)$ is
\begin{equation}
\concretization{\abstr{I}} = \{ (\vx,a) \mid
  \forall i\in\iota~ (\vx,i,a[i]) \in \abstr{I} \}
\end{equation}
The \emph{abstraction} of $I \subseteq \chi \times \arraytype{\iota}{\beta}$ is
\begin{equation}
\abstraction{I} = \{ (\vx,i,a[i]) \mid x \in \chi, i \in \iota \}
\end{equation}
\end{definition}

\begin{theorem}
$\alpha$ and $\gamma$ form a Galois connection
\begin{equation*}
 \parts{\chi \times \arraytype{\iota}{\beta}}
 \galois{\alpha}{\gamma} \parts{\chi \times (\iota \times \beta)}.
\end{equation*}
\end{theorem}

Our Horn rules are of the form $\forall \ve{y}~ \abstr{I_1}(\ve{f_1}(\ve{y})) \land \dots \land \abstr{I_1}(\ve{f_m}(\ve{y})) \land P(\ve{y}) \implies \abstr{I}_2(\ve{g}(\ve{y})$ ($\ve{y}$ is a vector of variables, $\ve{f_1},\dots,\ve{f_m}$ vectors of terms depending on $\ve{y}$, $P$ an arithmetic predicate over $\ve{y}$).
In other words, they impose in $\abstr{I}_2$ the presence of $\ve{g}(\ve{y})$ as soon as certain elements $\ve{f_1}(\ve{y}),\dots,\ve{f_m}(\ve{y})$ are found in~$\abstr{I_1}$.
Let $\abstr{I}_{2-}$ be the set of such imposed elements.
This Horn rule is said to be \emph{sound} if $\gamma(\abstr{I}_{2-})$ includes all states $(\vx',a')$ such that there exists $(\vx,a)$ in $\gamma(\abstr{I}_1)$ and $(\vx,a) \rightarrow_c (\vx',a')$.


\begin{lemma}
The forward abstract semantics of the read statement (Def.~\ref{def:read1}) is sound.
\end{lemma}

\begin{proof}
Let $(\vx,i,a) \in \gamma(\abstr{I}_1)$, that is
$\forall k \in \iota~ \abstr{I}_1(\vx,i,k,a[k])$.
Suppose $(\vx,i,a) \xrightarrow{v:=a[i]}_c (\vx,i,v,a)$, that is $v=a[i]$.
Let us now show that $(\vx,i,v,a) \in \concretization{\abstr{I}_{2-}}$, that is,
\begin{inparaenum}[i)]
\item for all $k \in \iota$ such that $k \neq i$, $\abstr{I}_1(x,i,i,v)$ and $\abstr{I}_1(x,i,k,a[k])$ both hold: both follow from $\forall k \in \iota~ \abstr{I}_1(\vx,i,k,a[k])$, and, for the first, from $v=a[i]$;
\item the case $k = i$ is also trivial.
\end{inparaenum}
\end{proof}

\begin{lemma}
The forward abstract semantics of the write statement (Def.~\ref{def:write1}) is sound.
\end{lemma}

\begin{proof}
Let $(\vx,i,a) \in \gamma(\abstr{I}_1)$, that is
$\forall k \in \iota~ \abstr{I}_1(\vx,i,k,a[k])$.
Suppose $(\vx,i,a) \xrightarrow{a[i]:=v}_c (\vx,i,v,a')$, that is, $a'[i]=v$ and for all $k \neq i$, $a'[k]=a[k]$
Let us now show that $(\vx,i,v,a') \in \concretization{\abstr{I}_{2-}}$, that is,  for all $k \in \iota$, either
\begin{inparaenum}[i)]
\item $\abstr{I}_1(\vx,i,v,k,a'_k)$ and $i\neq k$
\item $i=k$, $v=a'[k]$, and there exists $a_k$ such that $\abstr{I}_1(\vx,i,v,i,a_k)$.
\end{inparaenum}
Both cases are trivial.
\end{proof}

The following two lemma are also easily proved:
\begin{lemma}
The forward abstract semantics of array initialization (Def.~\ref{def:initialization1}) is sound.
\end{lemma}

\begin{lemma}
The forward abstract semantics of the scalar statements (Def.~\ref{def:scalar1}) is sound.
\end{lemma}

\begin{remark}
The scalar statements include ``killing'' dead variables (Def.~\ref{def:kill}).
Note that, contrary to many other abstractions, in ours, removing some variables may cause irrecoverable loss of precision on other variables \cite[Sec.~4.2]{Monniaux_Alberti_SAS2015}:
if $v$ is live, then one can represent $\forall k,~ a[k]=v$, which implies $\forall k_1,k_2~ a[k_1] = a[k_2]$ (constantness), but if $v$ is discarded, the constantness of $a$ is lost.
\end{remark}

\begin{theorem}
If $\abstr{I}_1,\dots,\abstr{I}_m$ are a solution of a system of Horn clauses sound in the above sense, then $\gamma(\abstr{I}_1),\dots,\gamma(\abstr{I}_m)$ are inductive invariants with respect to the concrete semantics~$\rightarrow_c$.
\end{theorem}

\begin{proof}
From the general properties of fixed points of monotone operators and Galois connections \cite{DBLP:journals/logcom/CousotC92}.
\end{proof}

\begin{figure}
\begin{center}
\begin{tikzpicture}[->, node distance=6.9em]
\node[state,initial by diamond] (init) {$\mathit{init}$};
\node[state] (loop) [right of=init] {$\mathit{loop}$};
\node[state] (write) [right of=loop] {$\mathit{write}$};
\node[state] (incr) [right of=write] {$\mathit{incr}$};
\node[state,accepting] (end) [below of=loop] {$\mathit{end}$};
\path (init) edge node[below, align=left] {$n > 0$\\$i := 0$} (loop);
\path (loop) edge node[above] {$i < n$} (write);
\path (write) edge node[above] {$a[i] := 42$} (incr);
\path (incr) edge[bend left=45] node[below] { $i := i+1$ } (loop);
\path (loop) edge node[left] {$i \geq n$} (end);
\end{tikzpicture}
\end{center}
\caption{Detailed control-flow graph for program~\ref{prog:array_fill1}}
\label{fig:array_fill1_detailed_cfg}
\end{figure}
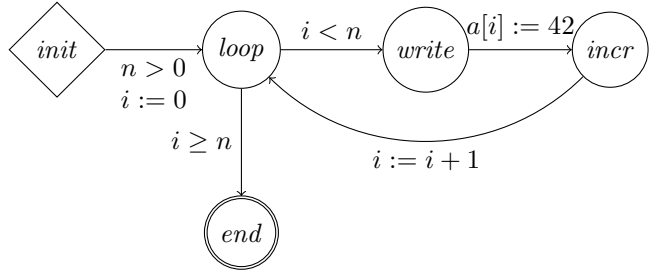

\begin{definition}[Property conversion]\label{def:property1}
A property ``at program point $p_i$, for all $\vx \in \chi$ and all $k \in \iota$, $\phi(\vx,k,a[k])$ holds'' (where $\phi$ is a formula, say over arithmetic) is converted into a Horn query $\forall \vx \in \chi~ \forall k \in \iota~ \phi(\vx,k,a_k)$.
\end{definition}

Our method for converting a scalar program into a system of Horn clauses over scalar variables is thus:
\begin{algo}\label{algo:construct_Horn}
\begin{enumerate}
\item Construct the control-flow graph of the program.
\item To each control point $p_i$, with vector of scalar variables $\vx_i$, associate a predicate $\abstr{I}_(\vx_i, k, a_k)$ in the Horn clause system (the vector of scalar variables may change from control point to control point).
\item For each transition of the program, generate Horn rules according to Def.~\ref{def:read1}, \ref{def:write1}, \ref{def:scalar1} as applicable (an initialization node does not need antecedents in its rule).
\item Generate Horn queries from desired properties according to Def.~\ref{def:property1}.
\end{enumerate}
\end{algo}

\begin{example*}[Ex.~\ref{ex:array_fill1}, continued]
Let us now apply the Horn abstract semantics from Definitions \ref{def:read1}, \ref{def:write1} and \ref{def:scalar1} to Program~\ref{prog:array_fill1}, following the detailed control-flow graph (Fig~\ref{fig:array_fill1_detailed_cfg});
in this case, $\alpha = \ZZ$, $\iota = \{ 0,\dots,n-1 \}$, $\chi = \ZZ$.
After a slight simplification of the Horn clauses, we obtain (Listing~\ref{Horn:array_fill1}):
{\small\begin{align}
\begin{aligned}
\forall n,k,a_k \in \ZZ~ 0 \leq k < n \implies \mathit{loop}(n,0,k,a_k)
\end{aligned}\label{rule:array_fill1_begin}\\
\begin{aligned}
\forall n,i,k,a_k \in \ZZ~ 0 \leq k < n \land i < n \land
\mathit{loop}(n,i,k,a_k)\\ \implies \mathit{write}(n,i,k,a_k)
\end{aligned}\rulespacing
\begin{aligned}
\forall n,i,k,a_k \in \ZZ~ 0 \leq k < n \land i \neq k \land
\mathit{write}(n,i,k,a_k)\\ \implies \mathit{incr}(n,i,k,a_k)
\end{aligned}\rulespacing
\begin{aligned}
\forall n,i,a_k \in \ZZ~ \land
\mathit{write}(n,i,i,a_k)\\ \implies \mathit{incr}(n,i,i,42)
\end{aligned}\rulespacing
\begin{aligned}
\forall n,i,k,a_k \in \ZZ~ 0 \leq k < n \land
\mathit{incr}(n,i,k,a_k)\\ \implies \mathit{loop}(n,i+1,k,a_k)
\end{aligned}\rulespacing
\begin{aligned}
\forall n,i,k,a_k \in \ZZ~ 0 \leq k < n \land i \geq n \land
\mathit{loop}(n,i,k,a_k)\\ \implies \mathit{end}(n,k,a_k)
\end{aligned}\label{rule:array_fill1_end}
\end{align}}

Finally, we add the postcondition (using Def.~\ref{def:property1}):
\begin{align}
\begin{aligned}
\forall n,k,a_k \in \ZZ~ 0 \leq k < n \land
\mathit{end}(n,i,k,a_k) \implies a_k = 42
\end{aligned}
\end{align}

Z3/PDR, \soft{Spacer} and \soft{Eldarica} (\verb+-splitClauses+), all unable to deal with the original array problem (Formulas \ref{equ:array_fill1_arrays_begin}--\ref{equ:array_fill1_arrays_end}) solve this problem quickly.
\end{example*}

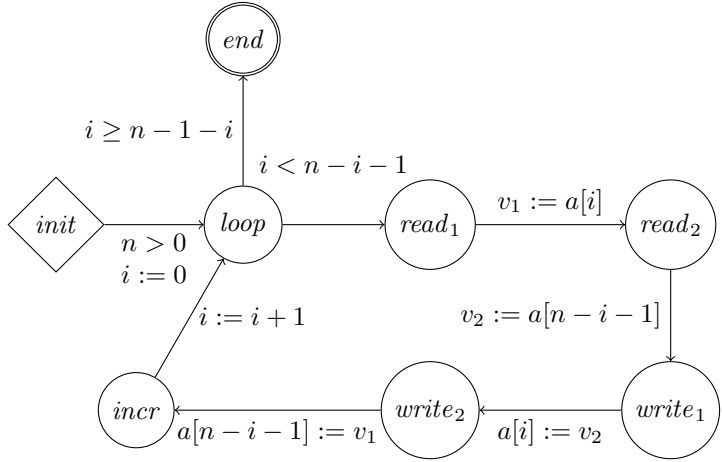
\begin{figure}
\begin{tikzpicture}[->, node distance=7em]
\node[state,initial by diamond] (init) {$\mathit{init}$};
\node[state] (loop) [right of=init] {$\mathit{loop}$};
\node[state] (read1) [right of=loop] {$\mathit{read}_1$};
\node[state] (read2) [node distance=9em, right of=read1] {$\mathit{read}_2$};
\node[state] (write1) [below of=read2] {$\mathit{write}_1$};
\node[state] (write2) [node distance=9em, left of=write1] {$\mathit{write}_2$};
\node[state] (incr) [left of=write2, node distance=11em] {$\mathit{incr}$};
\node[state,accepting] (end) [above of=loop] {$\mathit{end}$};
\path(init) edge node[below,align=left]{$n>0$\\$i:=0$} (loop);
\path(loop) edge node[above = 1.5em] {$i<n-i-1$} (read1);
\path(read1) edge node[above] {$v_1:=a[i]$} (read2);
\path(read2) edge node[left] {$v_2:=a[n-i-1]$} (write1);
\path(write1) edge node[below] {$a[i]:=v_2$} (write2);
\path(write2) edge node[below] {$a[n-i-1]:=v_1$} (incr);
\path(incr) edge node[right] {$i:=i+1$} (loop);
\path(loop) edge node[left]{$i\geq n-1-i$} (end);
\end{tikzpicture}

\caption{Array reversal}
\end{figure}

\begin{example}\label{ex:array_reverse_once1}
Consider now an array reversal procedure:
\begin{lstlisting}[label=prog:array_reverse_once1,caption={Array reversal}]
void reverse(int n, int a[n]) {
  int i = 0;
  while(true) {
    int j = n-1-i;
    if (i >= j) break;
    int v1 = a[i], v2 = a[j];
    a[i] = v2; a[j] = v1;
    i = i+1;
  }
}
\end{lstlisting}

In order to prove that the final array is the reversal of the initial array, we keep a copy of the initial array as~$b$.
Applying the rules for the forward abstraction (Def.~\ref{def:read1}, \ref{def:write1}, \ref{def:scalar1}), and some simple simplifications (removal of dead variables, propagation of $j = n-i-1$), we obtain the Horn clauses (Listing~\ref{Horn:array_reverse_once1}):
{\small\begin{align}
\begin{aligned}
\forall n,k,a_k,l,b_l \in \ZZ~
   \mathit{init}(n,k,a_l,l,b_l) \\\implies \mathit{loop}(n,0,k,a_k,l,b_l)
\end{aligned}\rulespacing
\begin{aligned}
\forall n,i,k,a_k,l,b_l \in \ZZ~
   \mathit{loop}(n,i,k,a_k,l,b_l) \\ \land i< n-i-1 \implies
   \mathit{read}_1(n,i,k,a_k,l,b_l)
\end{aligned}\rulespacing
\begin{aligned}
\forall n,i,v_1,v_2,k,a_k,l,b_l \in \ZZ~
   \mathit{read}_1(n,i,k,a_k,l,b_l) \\ \land i \neq k \land
   \mathit{read}_1(n,i,i,v_1,l,b_l) \\ \implies
   \mathit{read}_2(n,i,v_1,k,a_k,l,b_l)
\end{aligned}\rulespacing
\begin{aligned}
\forall n,i,v_1,v_2,k,l,b_l \in \ZZ~
   \mathit{read}_1(n,i,i,v_1,b_l) \\ \implies
   \mathit{read}_2(n,i,v_1,i,v_1,l,b_l)
\end{aligned}\rulespacing
\begin{aligned}
\forall n,i,v_1,v_2,k,a_k,l,b_l \in \ZZ~
   \mathit{read}_2(n,i,v_1,v_2,k,a_k,l,b_l) \\\land n-1-i \neq k \land
   \mathit{read}_2(n,i,v_1,n-i-1,v_2,l,b_l)\\ \implies
   \mathit{write}_1(n,i,v_1,v_2,k,a_k,l,b_l)
\end{aligned}\rulespacing
\begin{aligned}
\forall n,i,v_1,v_2,l,b_l \in \ZZ~
   \mathit{read}_2(n,i,v_1,v_2,n-1-i,v_2,l,b_l) \\ \implies
   \mathit{write}_1(n,i,v_1,v_2,n-1-i,v_2,l,b_l)
\end{aligned}\rulespacing
\begin{aligned}
\forall n,i,v_1,v_2,k,a_k,l,b_l \in \ZZ~
   \mathit{write}_1(n,i,v_1,v_2,k,a_k,l,b_l) \\ \land i\neq k \implies
   \mathit{write}_2(n,i,v_1,k,a_k,l,b_l)
\end{aligned}\rulespacing
\begin{aligned}
\forall n,i,v_1,v_2,k,a_k,l,b_l \in \ZZ~
   \mathit{write}_1(n,i,v_1,v_2,k,a_k,l,b_l) \\\implies
   \mathit{write}_2(n,i,v_1,i,v_2,l,b_l)
\end{aligned}\rulespacing
\begin{aligned}
\forall n,i,v_1,k,a_k,l,b_l \in \ZZ~
   \mathit{write}_2(n,i,v_1,k,a_k,l,b_l) \\\land n-1-i\neq k \implies
   \mathit{incr}(n,i,k,a_k,l,b_l)
\end{aligned}\rulespacing
\begin{aligned}
\forall n,i,v_1,k,a_k,l,b_l \in \ZZ~
   \mathit{write}_2(n,i,v_1,k,a_k,l,b_l) \\\implies
   \mathit{incr}(n,i,n-1-i,v_1,l,b_l)
\end{aligned}\rulespacing
\begin{aligned}
\forall n,i,k,a_k,l,b_l \in \ZZ~
   \mathit{incr}(n,i,k,a_k,l,b_l) \\\implies
   \mathit{loop}(n,i+1,k,a_k,l,b_l)
\end{aligned}\rulespacing
\begin{aligned}
\forall n,i,k,a_k,l,b_l \in \ZZ~
   i\geq n-i-1 \land \mathit{loop}(n,i,k,a_k,l,b_l) \\\implies
   \mathit{end}_1(n,k,a_k,l,b_l)
\end{aligned}
\end{align}
}

We specify that, initially, $a[k]=b[k]$ for all legal index $k$ as:
\begin{align}
\forall n,k,a_k \in \ZZ~ 0 \leq k < n \implies \mathit{init}(n,k,a_k,k,a_k)\\
\begin{aligned}
\forall n,k,a_k,l,a_l \in \ZZ~ 0 \leq k < n \land 0 \leq l < n \land k \neq l
  \\\implies \mathit{init}(n,k,a_k,l,b_l)
\end{aligned}
\end{align}

We finally specify that the final array is the reversal of the initial:
\begin{align}
\begin{aligned}
\forall n,i,a_k,b_j \in \ZZ~ 0 \leq i < n \land
  \mathit{end}(n, i, a_k, n-1-k, b_j)\\
  \implies a_k = b_j
\end{aligned}
\end{align}

\noindent\soft{Z3/PDR} solves this problem within 4~s, \soft{Spacer} takes 5~min
\end{example}

\begin{example}
Consider now the problem of finding the minimum of an array slice $a[l \dots h-1]$, with value $b=a[p]$:
\begin{lstlisting}[label=prog:find_minimum,caption={Find minimum in an array slice}]
void find_minimum(int n, int a[n], int l, int h){
  int p = l, b = a[l], i = l+1;
  while(i < h) {
    int v = a[i];
    if (v < b) {
      b = v;
      p = i;
    }
    i = i+1;
  }
}
\end{lstlisting}

Again, we encode the abstraction of the statements (Def.~\ref{def:read1}, \ref{def:write1}, \ref{def:scalar1}) as Horn clauses (Listing~\ref{Horn:find_minimum}).
At the end we have a predicate $\mathit{end}(l,h,p,b,k,a[k])$ on which we impose the properties
\begin{align}
\begin{aligned}
\forall l,h,p,b,a_p~ \mathit{end}(l,h,p,b,p,a_p) \implies b = a_p
   \label{formula:find_minimum:found}
\end{aligned} \\
\begin{aligned}
\forall l,h,p,b,k,a_p,a_k~
  l \leq k < h \land \mathit{end}(l,h,p,b,k,a_k) \\ \implies b \leq a_k
   \label{formula:find_minimum:is_minimum}
\end{aligned}
\end{align}
Rule~\ref{formula:find_minimum:found} imposes the postcondition $b = a[p]$, Rule~\ref{formula:find_minimum:is_minimum} imposes the postcondition $\forall k~ l \leq k < h \implies b \leq a[k]$.
Again, \soft{Z3/PDR} and \soft{Spacer} solve this Horn system (but not \soft{Eldarica}).
\end{example}

The kind of relationship that can be inferred between loop indices, array indices and array contents is limited only by the capabilities of the Horn solver, as shown in the following example:

\begin{example}
Consider for instance this array fill where $a[i]$ gets $i \bmod 2$:
\begin{lstlisting}[label=prog:array_fill1_even_odd,caption={Fill 1D-array with even/odd values}]
void array_fill1_even_odd(int n, int a[n]) {
  int i = 0;
  while(i < n) {
    a[i] = i & 1;
    i = i+1;
  }
}
\end{lstlisting}

The abstract semantics is Formulas \ref{rule:array_fill1_begin}--\ref{rule:array_fill1_end} except that the constant 42 is replaced by $i \bmod 2$.
We wish to prove postconditions
\begin{align}
\forall k~ 0 \leq 2k < n \implies a[2k] = 0\\
\forall k~ 0 \leq 2k+1 < n \implies a[2k+1] = 1
\end{align}
which get translated into Horn clauses
\begin{align}
\forall k~ \mathit{end}(n,2k,a_x) \implies a_x = 0\\
\forall k~ \mathit{end}(n,2k+1,a_x) \implies a_x = 1
\end{align}
\soft{Spacer} solves this problem (Listing~\ref{Horn:array_fill1_even_odd}) instantaneously, while \soft{Z3/PDR} cannot solve it.
We suppose this is because Z3/PDR cannot infer interpolants depending on divisibility predicates.
\end{example}

We have made no assumption regarding the nature of the indexing variable: we used integers because arrays indexed by an integer range are a very common kind of data structure, but really it can be any type supported by the Horn clause solver, e.g. rationals:

\begin{example}
Consider the following program, handling a mutable map \lstinline|a[]| from the rationals to the integers, initialized to 0:
\begin{lstlisting}
a[1] = 10;   a[2] = 20;   a[3] = 30;
\end{lstlisting}

We can encode it as (Listing~\ref{Horn:real_indexed_maps}):
\begin{align}
\forall x \in \QQ~ \mathit{init}(x, 0)\\
\forall x \in \QQ~ \forall a_x \in \ZZ~
  \mathit{init}(x, a_x) \land x\neq 1 \implies \mathit{w}_1(x, a_x)\\
\forall a_x \in \ZZ~
  \mathit{init}(1, ax) \implies \mathit{w}_1(1, 10)\\
\forall x \in \QQ~ \forall a_x \in \ZZ~
  \mathit{w}_1(x, a_x) \land x\neq 2 \implies \mathit{w}_2(x, a_x)\\
\forall a_x \in \ZZ~
  \mathit{w}_1(2, ax) \implies \mathit{w}_2(2, 20)\\
\forall x \in \QQ~ \forall a_x \in \ZZ~
  \mathit{w}_2(x, a_x) \land x\neq 3 \implies \mathit{exit}(x, a_x)\\
\forall a_x \in \ZZ~
  \mathit{w}_2(3, ax) \implies \mathit{exit}(3, 30)
\end{align}

The postcondition $\forall x\in \QQ~ a[x] \geq 0$, encoded as $\forall x\in \QQ~ a_x \geq 0$, is easily proved by \soft{Z3/PDR} and \soft{Spacer}.
\end{example}

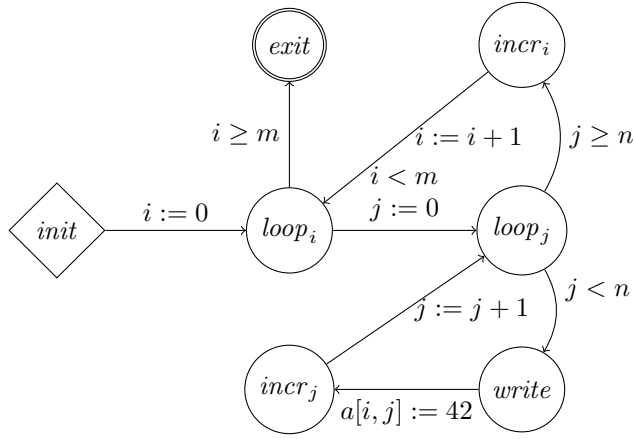
\begin{figure}
\begin{center}
\begin{tikzpicture}[->, node distance=8.7em]
\node[state, initial by diamond] (init) {$\mathit{init}$};
\node[state, right of=init] (loopi) {$\mathit{loop}_i$};
\node[state, right of=loopi] (loopj) {$\mathit{loop}_j$};
\node[state, below of=loopj, node distance=6em] (write) {$\mathit{write}$};
\node[state, left of=write] (incrj) {$\mathit{incr}_j$};
\node[state, above of=loopj, node distance=7em] (incri) {$\mathit{incr}_i$};
\node[state, accepting, above of=loopi, node distance=7em] (exit)  {$\mathit{exit}$};
\path(init) edge node[above] {$i:=0$} (loopi);
\path(loopi) edge node[above, align=left] {$i<m$ \\ $j:=0$} (loopj);
\path(loopj) edge[bend left] node[above right] {$j<n$} (write);
\path(write) edge node[below] {$a[i,j] := 42$} (incrj);
\path(incrj) edge node[right] {$j := j+1$} (loopj);
\path(incri) edge[right] node {$i := i+1$} (loopi);
\path(loopj) edge[bend right] node[right] {$j\geq n$} (incri);
\path(loopi) edge node[left] {$i\geq m$} (exit);
\end{tikzpicture}
\end{center}

\caption{Fill 2D array}
\label{fig:array_fill2}
\end{figure}

\emph{Matrices} are bidimensional arrays, that is, arrays indexed by two integers $x$ and $y$: $0 \leq x < m$, $0 \leq y < n$ for a $m \times n$ arrays.
More generally, arrays can be defined for an arbitrary number $d$ of dimensions.
Everything that we have seen so far applies when the type $\iota$ of the indexing variable is a subset of $\ZZ^d$ (e.g. for $d=2$, $\iota = \{ (x,y) \mid 0 \leq x < m \land 0 \leq y < n \}$).
We may therefore apply directly what precedes and generate Horn clauses referring to pairs of indices $(x,y)$.
Since not every solver supports these, one may instead use two indices $x$ and $y$: a comparison $(x_1,y_1) = (x_2,y_2)$ is expressed as $x_1=x_2 \land y_1=y_2$.

\begin{example}
The following program fills a $m \times n$ matrix:
\begin{lstlisting}[caption={Fill 2D-matrix},label={prog:array_fill2}]
void array_fill2(int m, int n, int a[m][n]) {
  int i = 0;
  while(i < m) {
    int j = 0;
    while(j < n) {
      a[i][j] = 42;
      j = j+1;
    }
    i = i+1;
  }
}
\end{lstlisting}

It gets encoded as (Listing~\ref{Horn:array_fill2}, Fig.~\ref{fig:array_fill2}):
{\small\begin{align}
\begin{aligned}
\forall m,n,x,y,a_{xy} \in \ZZ~
  0 \leq x < m \land 0 \leq y < n \\ \implies
  \mathit{init}(m, n, x, y, a_{xy})
\end{aligned}\rulespacing
\begin{aligned}
\forall m,n,x,y,a_{xy} \in \ZZ~
  \mathit{init}(m, n, x, y, a_{xy}) \\ \implies
  \mathit{loop}_i(m, n, 0, x, y, a_{xy})
\end{aligned}\rulespacing
\begin{aligned}
\forall m,n,i,x,y,a_{xy} \in \ZZ~
  \mathit{loop}_i(m, n, i, x, y, a_{xy}) \land i<m \\ \implies
  \mathit{loop}_j(m, n, i, 0, x, y, a_{xy})
\end{aligned}\rulespacing
\begin{aligned}
\forall m,n,i,x,y,a_{xy} \in \ZZ~
  \mathit{loop}_i(m, n, i, x, y, a_{xy}) \land i\geq m \\ \implies
  \mathit{exit}(m, n, x, y, a_{xy})
\end{aligned}\rulespacing
\begin{aligned}
\forall m,n,i,j,x,y,a_{xy} \in \ZZ~
  \mathit{loop}_j(m, n, i, j, x, y, a_{xy}) \land j<n \\ \implies
  \mathit{write}(m, n, i, j, x, y, a_{xy})
\end{aligned}\rulespacing
\begin{aligned}
\forall m,n,i,j,x,y,a_{xy} \in \ZZ~
  \mathit{loop}_j(m, n, i, j, x, y, a_{xy}) \land j\geq n \\ \implies
  \mathit{incr}_i(m, n, i, x, y, a_{xy})
\end{aligned}\rulespacing
\begin{aligned}
\forall m,n,i,j,x,y,a_{xy},a_{ij} \in \ZZ \\
  \mathit{write}(m, n, i, j, x, y, a_{xy})  \land
  \mathit{write}(m, n, i, j, i, j, a_{ij}) \\ \land
  \land (i\neq x \lor j\neq y) \implies
  \mathit{incr}_j(m, n, i, j, x, y, a_{xy})
\end{aligned}\rulespacing
\begin{aligned}
\forall m,n,i,j,a_{ij} \in \ZZ~
  \mathit{write}(m, n, i, j, i, j, a_{ij}) \\ \implies
  \mathit{incr}_j(m, n, i, j, i, j, 42)
\end{aligned}\rulespacing
\begin{aligned}
\forall m,n,i,j,x,y,a_{xy} \in \ZZ~
  \mathit{incr}_j(m, n, i, j, x, y, a_{xy}) \\ \implies
  \mathit{loop}_j(m, n, i, j+1, x, y, a_{xy})
\end{aligned}\rulespacing
\begin{aligned}
\forall m,n,i,x,y,a_{xy} \in \ZZ~
  \mathit{incr}_i(m, n, i, x, y, a_{xy}) \\ \implies
  \mathit{loop}_i(m, n, i+1, x, y, a_{xy})
\end{aligned}
\end{align}
}

Again, we can prove that $\forall m, n, x, y, a_{xy} \in \ZZ,
\mathit{exit} \implies a_{xy} = 42$; otherwise said, finally,
$\forall x,y~ a[x,y]=42$.
\end{example}

\section{Sortedness}
\label{sec:sortedness}
The Galois connection of Def.~\ref{def:Galois1} expresses relations of the form $\forall k \in \iota~ \phi(\vx, k, a[k])$ where $\vx$ are variables from the program, $a$ a map and $k$ an index into the map $a$; in other words, relations between each array element individually and the rest of the variables.
It cannot express properties such as sortedness, which link \emph{two} array elements: $\forall k_1,k_2 \in \iota~ k_1 < k_2 \implies a[k_1] \leq a[k_2]$.
Let us now see an abstraction with two ``distinguished cells'', capable of representing such properties:

\begin{definition}
The concretization with two indices of $\abstr{I} \subseteq \chi \times (\iota \times \beta)^2$ is
\begin{equation}
\concretization[2]{\abstr{I}} = \{ (\vx,a) \mid
  \forall k_1,k_2\in\iota~ (\vx,k_1,a[k_1],k_2,a[k_2]) \in \abstr{I} \}
\end{equation}
The abstraction with two indices of $I \subseteq \chi \times \arraytype{\iota}{\beta}$ is
\begin{equation}
\abstraction[2]{I} = \{ (\vx,k_1,a[k_1],k_2,a[k_2]) \mid x \in \chi, k_1,k_2 \in \iota \}
\end{equation}
\end{definition}

\begin{theorem}
$\alpha_2$ and $\gamma_2$ form a Galois connection
\begin{equation*}
 \parts{\chi \times \arraytype{\iota}{\beta}}
 \galois{\alpha_2}{\gamma_2} \parts{\chi \times (\iota \times \beta)^2}.
\end{equation*}
\end{theorem}

With respect to implementation efficiency, it may be preferable to break this symmetry between indices $k_1$ and $k_2$ by imposing $k_1 \leq k_2$ for some total order. One then gets:

\begin{definition}
The concretization with two ordered indices of $\abstr{I} \subseteq \chi \times (\iota \times \beta)^2$ is
\begin{equation}
\concretization[2\leq]{\abstr{I}} = \{ (\vx,a) \mid
  \forall k_1 \leq k_2\in\iota~ (\vx,k_1,a[k_1],k_2,a[k_2]) \in \abstr{I} \}
\end{equation}
The abstraction with two indices of $I \subseteq \chi \times \arraytype{\iota}{\beta}$ is
\begin{equation}
\abstraction[2\leq]{I} = \{ (\vx,k_1,a[k_1],k_2,a[k_2]) \mid x \in \chi, k_1 \leq k_2 \in \iota \}
\end{equation}
\end{definition}

\begin{theorem}
$\alpha_{2\leq}$ and $\gamma_{2\leq}$ form a Galois connection
\begin{multline*}
 \parts{\chi \times \arraytype{\iota}{\beta}}
 \galois{\alpha_{2\leq}}{\gamma_{2\leq}}\\
 \parts{\{(x,k_1,v_1,k_2,v_2) \mid x \in \chi,~ k_1 \leq k_2 \in \iota,~ v_1, v_2 \in \beta \}}.
\end{multline*}
\end{theorem}

\begin{definition}[Read statement, two indices $k_1 \leq k_2$]\label{def:read2}
The abstraction of $v := a[i]$ is:
{\small\begin{align}
\begin{aligned}
\forall \vx \in \chi~ \forall i,k_1,k_2\in\iota~ \forall v,a_{k_1},a_{k_2} \in \beta\\
  \abstr{I}_1(\vx, i, k_1, a_{k_1}, k_2, a_{k_2}) \land
  \abstr{I}_1(\vx, i, i, v, k_2, a_{k_2}) \\ \land
  k_1 \neq i \land i < k_2 \implies
  \abstr{I}_2(\vx, i, v, k_1, a_{k_1}, k_2, a_{k_2})
\end{aligned} \label{rule:read2_begin}\\
\begin{aligned}
\forall \vx \in \chi~ \forall i,k_1,k_2\in\iota~ \forall v,a_{k_1},a_{k_2} \in \beta \\
  \abstr{I}_1(\vx, i, k_1, a_{k_1}, k_2, a_{k_2}) \land
  \abstr{I}_1(\vx, i, k_1, a_{k_1}, i, v) \\ \land
  k_2 \neq i \land k_1 < i \implies
  \abstr{I}_2(\vx, i, v, k_1, a_{k_1}, k_2, a_{k_2})
\end{aligned}\rulespacing
\begin{aligned}
\forall \vx \in \chi~ \forall i,k_2\in\iota~ \forall v,a_{k_2} \in \beta \\
  \abstr{I}_1(\vx, i, i, v, k_2, a_{k_2}) \land i < k_2 \implies
  \abstr{I}_2(\vx, i, i, v, k_2, a_{k_2})
\end{aligned}\rulespacing
\begin{aligned}
\forall \vx \in \chi~ \forall i,k_1\in\iota~ \forall v,a_{k_1} \in \beta \\
  \abstr{I}_1(\vx, i, k_1, a_{k_1}, i, v) \land k_1 \leq i \implies
  \abstr{I}_2(\vx, i, k_1, a_{k_1}, i, v)
\end{aligned}\label{rule:read2_end}
\end{align}} 
\end{definition}

\begin{definition}[Write statement, two indices $k_1 \leq k_2$]\label{def:write2}
The abstraction of $a[i]:=v$ is:
{\small\begin{align}
\begin{aligned}
\forall \vx \in \chi~ \forall i,k_1,k_2\in\iota~ \forall v,a_{k_1},a_{k_2} \in \beta\\
  \abstr{I}_1(\vx, i, v, k_1, a_{k_1}, k_2, a_{k_2})
  \land i \neq k_1 \land i \neq k_2 \\ \implies
  \abstr{I}_2(\vx, i, v, k_1, a_{k_1}, k_2, a_{k_2})
\end{aligned}\rulespacing
\begin{aligned}
\forall \vx \in \chi~ \forall i, k_2\in\iota~
  \forall v,a_{k_1},a_{k_2} \in \beta~~ \land i \neq k_2 \\
  \land \abstr{I}_1(\vx, i, v, i, a_{k_1}, k_2, a_{k_2})
  \implies
  \abstr{I}_2(\vx, i, v, i, v, k_2, a_{k_2})
\end{aligned}\rulespacing
\begin{aligned}
\forall \vx \in \chi~ \forall i,k_1\in\iota~
  \forall v,a_{k_1},a_{k_2} \in \beta`~ \land i \neq k_1 \\
  \land \abstr{I}_1(\vx, i, v, k_1, a_{k_1}, i, a_{k_2})
  \implies
  \abstr{I}_2(\vx, i, v, k_1, a_{k_1}, i, v)
\end{aligned}\rulespacing
\begin{aligned}
\forall \vx \in \chi~ \forall i\in\iota~ \forall v,a_k \in \beta\\
  \abstr{I}_1(\vx, i, v, i, a_k, i, a_k)
  \implies
  \abstr{I}_2(\vx, i, v, i, v, i, v)
\end{aligned}
\end{align}}
\end{definition}

\begin{lemma}
The abstract forward semantics of the read statement, with two indices $k_1 \leq k_2$ (Def.~\ref{def:read2}) is a sound abstraction of the concrete semantics given in Def.~\ref{def:read1}.
\end{lemma}

\begin{proof}
Let $k_1 \leq k_2 \in \iota$ and $i \in \iota$, $\iota$ being totally ordered.
Then one necessarily falls in one of the following cases, corresponding to rules \ref{rule:read2_begin}--\ref{rule:read2_end}: 
\begin{inparaenum}[i)]
\item $k_1 \neq i \land i < k_2$
\item $k_2 \neq i \land k < i$
\item $k_1 = i \land i < k_2$
\item $k_2 = i \land k_1 \leq i$
\end{inparaenum}
It is easy to see that, for each of these cases, the corresponding rule is sound.
\end{proof}

\begin{lemma}
The abstract forward semantics of the write statement, with two indices $k_1 \leq k_2$ (Def.~\ref{def:write2}) is a sound abstraction of the concrete semantics given in Def.~\ref{def:write1}.
\end{lemma}

\begin{proof}
Similarly, one is necessarily in one of the four exclusive cases, each corresponding to a rule easily proved to be sound:
\begin{inparaenum}[i)]
\item $i \neq k_1 \land i \neq k_2$
\item $i = k_1 \land i \neq k_2$
\item $i \neq k_1 \land i = k_2$
\item $i = k_1 = k_2$.
\end{inparaenum}
\end{proof}

It is possible to mix abstractions with one or two ``distinguished cells'' within the same problem. Let us see how to convert between the two:

\begin{definition}\label{def:expansion1to2}
Let $\abstr{I}_1$ be an abstraction according to $\galois{\gamma}{\alpha}$; we wish to expand it to an abstraction $\abstr{I}_2$ according to $\galois{\gamma_{2\leq}}{\alpha_{2\leq}}$.

\begin{align}
\begin{aligned}
\forall \vx \in \chi~ \forall k_1, k_2\in\iota~ \forall \beta_1,\beta_2 \in \beta\\
  \abstr{I}_1(\vx, k_1, a_{k_1}) \land \abstr{I}_1(\vx, k_2, a_{k_2}) 
  \land k_1 < k_2
  \\ \implies
  \abstr{I}_2(\vx, k_1,  a_{k_1}, k_2, a_{k_2})
\end{aligned}\rulespacing
\begin{aligned}
\forall \vx \in \chi~ \forall k \in\iota~ \forall \beta \in \beta\\
  \abstr{I}_1(\vx, k) \implies
  \abstr{I}_2(\vx, k,  a_k, k, a_k)
\end{aligned}
\end{align}
\end{definition}

\begin{remark}[Initialization]\label{def:initialization2}
By coalescing these rules with initialization for one ``distinguished cell'' (Def.~\ref{def:initialization1}), we obtain a direct abstraction of initialization for the case of an array indexed by $0 \dots n-1$, which we use in our examples:

\begin{align}
\begin{aligned}
\forall \vx \in \chi~ \forall n \in \ZZ~ \forall k_1, k_2\in\ZZ~ \forall \beta_1,\beta_2 \in \beta\\
  0 \leq k_1 < k_2 < n \land
  \abstr{I}_1(\vx) \implies
  \abstr{I}_2(\vx, k_1,  a_{k_1}, k_2, a_{k_2})
\end{aligned}\rulespacing
\begin{aligned}
\forall \vx \in \chi~ \forall k \in\ZZ~ \forall \beta \in \beta\\
  0 \leq k < n \land
  \abstr{I}_1(\vx) \implies  \abstr{I}_2(\vx, k,  a_k, k, a_k)
\end{aligned}
\end{align}
\end{remark}
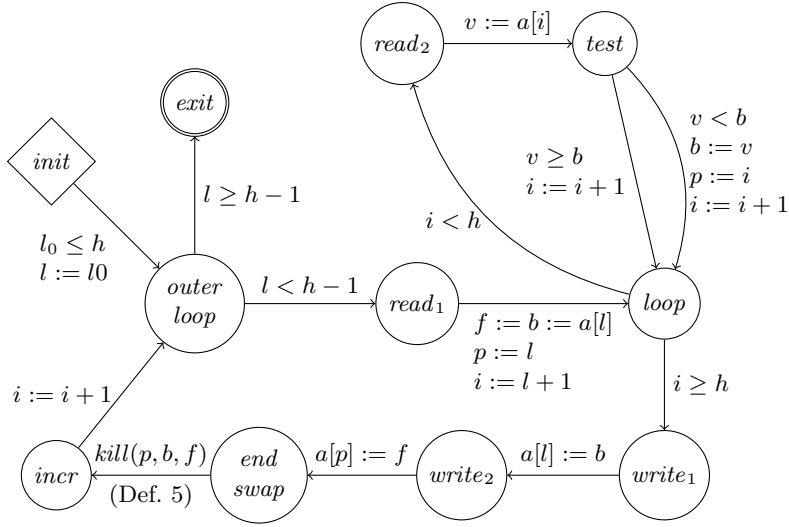
\begin{figure}
{\small
\begin{tikzpicture}[->, node distance=8.2em]
\node[state,initial by diamond] (init) {$\mathit{init}$};
\node[state,below right of=init, align=center] (outerloop) {$\mathit{outer}$\\$\mathit{loop}$};
\node[state,accepting,above of=outerloop] (exit) {$\mathit{exit}$};
\node[state,right of=outerloop, node distance=9em] (read1) {$\mathit{read}_1$};
\node[state,node distance=10em, right of=read1] (loop) {$\mathit{loop}$};
\node[state,above left of=loop,node distance=15em] (read2) {$\mathit{read}_2$};
\node[state,right of=read2] (test) {$\mathit{test}$};
\node[state,below of=loop, node distance=7em] (write1) {$\mathit{write}_1$};
\node[state,left of=write1] (write2) {$\mathit{write}_2$};
\node[state,left of=write2, align=center] (endswap) {$\mathit{end}$\\$\mathit{swap}$};
\node[state,left of=endswap] (incr) {$\mathit{incr}$};
\path (init) edge node[below left,align=left]{$l_0 \leq h$ \\ $l:=l0$} (outerloop);
\path (outerloop) edge node[above]{$l < h-1$} (read1);
\path (outerloop) edge node[right]{$l \geq h-1$} (exit);
\path (read1) edge node[align=left,below]{$f := b := a[l]$ \\ $p:=l$ \\ $i:=l+1$} (loop);
\path (loop) edge[bend left] node[left]{$i < h$} (read2);
\path (read2) edge node[above]{$v := a[i]$} (test);
\path (test) edge[bend left] node[right,align=left]{$v < b$ \\ $b := v$ \\ $p:=i$ \\ $i := i+1$} (loop);
\path (test) edge node[left,align=left]{$v \geq b$ \\ $i := i+1$} (loop);
\path (loop) edge node[right]{$i \geq h$} (write1);
\path (write1) edge node[above]{$a[l] := b$} (write2);
\path (write2) edge node[above]{$a[p] := f$} (endswap);
\path (endswap) edge node[above]{$\mathit{kill}(p,b,f)$}
   node[below]{(Def.~\ref{def:kill})}
   (incr);
\path (incr) edge node[left]{$i:=i+1$} (outerloop);
\end{tikzpicture}}

\caption{Selection sort}\label{fig:selection_sort}
\end{figure}

\begin{example}[Selection sort]\label{ex:selection_sort_sortedness}
Selection sort finds the least element in $a[l \dots h-1]$ (using Prog.~\ref{prog:find_minimum} as its inner loop) and swaps it with $a[l]$, then sorts $a[l+1,h-1]$. At the end, $a[l_0\dots h-1]$ is sorted, where $l_0$ is the initial value of~$l$.

\begin{lstlisting}[caption={Selection sort},label=prog:selection_sort]
void selection_sort(int l0, int h, int a[]) {
  int l = l0;
  while (l < h-1) {
    int p = l, b = a[l], f = b, i = l+1;
    while(i < h) {
      int v = a[i];
      if (v < b) {
	b = v;
	p = i;
      }
      i = i+1;
    }
    a[l] = b;
    a[p] = f;    
    l = l+1;
  }
}
\end{lstlisting}

Using the control points from Fig.~\ref{fig:selection_sort}, and the rules for the read (Def.~\ref{def:read2}) and write (Def.~\ref{def:write2}) statements, we write the abstract forward semantics of this program as a system of Horn clauses (Listing~\ref{Horn:selection_sort_sortedness}).

We wish to prove that, at the end, $a[l_0,h-1]$ is sorted: at the $\mathit{exit}$ node,
\begin{equation}
\forall l_0 \leq k_1 < k_2 < h~~ a[k_1] \leq a[k_2]
\end{equation}

This is expressed as the final condition
\begin{align}
\begin{aligned}
\forall l_0,h,k_1,a_{k_1},k_2,a_{k_2}~
  l_0 \leq k < k_2 < h \\ \land \mathit{exit}(l_0, h, k_1, a_{k_1}, k_2, a_{k_2})
  \implies a_{k_1} \leq a_{k_2}
\end{aligned}
\end{align}

\soft{Spacer} solves the resulting system of Horn clauses in 8~min.
\end{example}

We have thus, fully automatically, proved that the output of selection sort is truly sorted (in Example~\ref{ex:selection_sort_multiset} we shall see how to prove that the multiset of elements in the output is the same as in the input).

One may be concerned about an analysis time of 6 minutes for a 17-line program.
In our experience, the average undergraduate student taking a course in program verification and asked to provide inductive invariants (in the Floyd-Hoare sense; say, as annotations for a tool such as Frama-C) to prove that property takes longer time to provide them.
In particular, this invariant for the outer loop is somewhat non trivial:
\begin{equation}
\forall k_1,k_2~ l_0 \leq k_1 < l \land k_1 \leq k_2 < h
\implies a[k_1] \leq a[k_2]
\end{equation}
This invariant can be expressed in our system of Horn clauses as:
\begin{align}
\begin{aligned}
\forall l_0,l,h,k_1,a_{k_1},k_2,a_{k_2} \in \ZZ~
  l_0 \leq k_1 < l \land k_1 \leq k_2 < h \\ \land
  \mathit{outerloop}(l_0,l,h,k_1,a_{k_1},k_2,a_{k_2})
  \implies a_{k_1} \leq a_{k_2}
\end{aligned}\label{formula:outerloop_invariant}
\end{align}
If this invariant is added to the problem as an additional query to prove, \soft{Spacer} solves the problem in 1 second!
It could seem counter-intuitive that a solver would take less time to solve a problem with an additional constraint; but this constraint expresses an invariant necessary to prove the solution, and thus nudges the solver towards the solution.

Our approach is therefore flexible: if a solver fails to prove the desired property on its own, it is possible to help it by providing partial invariants. This is a less tedious approach than having to provide full invariants at every loop header, as common in assisted Floyd-Hoare proofs.

\section{Sets and multisets}
Our abstraction for maps may be used to abstract (multi)sets.

\subsection{Simple sets and multisets}\label{sec:simple_sets}
Many programming languages provide libraries for computing over sets or multisets of elements. One should reason on programs using these libraries by using the set-theoretic, high-level specification of their interface, as opposed to internal implementation details.

Remark, again, that we have made no assumption on the set of indices $\iota$ (except, occasionally, that is endowed with a total order, but that assumption may be dispensed from).
A subset of $\iota$ is just a map from $\iota$ to the Booleans, a multiset a map from $\iota$ to the natural numbers.
Testing the membership of one item $k \in \iota$ therefore just amounts to an array read~$a[k]$, forcing membership or non-membership just amounts to a write.

A single (multi)set $a$ is abstracted as a set of pairs $(k,a[k])$. If one has several (multi)sets $a,b,c$, one may either abstract them with separate indices $(i,a[i],j,a[j],k,a[k])$, or with a common index $(k,a[k],b[k],c[k])$. This last option is less expressive, but simpler, and is often sufficient.

\begin{definition}[(Multi)set union]
The operation $a:=\mathit{union}(b,c)$ is abstracted as:
\begin{align}
\forall \vx \in \chi~ \forall k \in \iota~
  \abstr{I}_1(\vx, k, a_k, b_k, c_k) \implies
  \abstr{I}_2(\vx, k, b_k \lor c_k, b_k, c_k)
\end{align}
(For multiset, replace $\lor$ by~$+$.)
\end{definition}

\begin{definition}[Set intersection]
The operation $a:=\mathit{intersection}(b,c)$ is abstracted as:
\begin{align}
\forall \vx \in \chi~ \forall k \in \iota~
  \abstr{I}_1(\vx, k, a_k, b_k, c_k) \implies
  \abstr{I}_2(\vx, k, b_k \land c_k, b_k, c_k)
\end{align}
\end{definition}

If operations such as ``get the (min/max)imal element'' are to be abstracted precisely, then one can enrich the abstraction by adding tracking variables $l$ and $h$ for the minimal and maximal elements, and updating them accordingly.
In the case of sets of integers, such tracking variables may be used to implement the ``for each'' iterator: iterate $i$ from $l$ to $h$ and test whether $i$ is in the set.

\subsection{Multiset of elements in an array}
\label{sec:multiset_of_elements}
In Example~\ref{ex:selection_sort_sortedness}, we showed how to prove that the output of selection sort is sorted. This is not enough for functional correctness: we also have to prove that the output is a permutation of the input, or, equivalently, that the multiset of elements in the output array is the same as that in the input array.

Let us remark that it is easy to keep track, in an auxiliary map, of the number $\hash a(x)$ of elements of value $x$ in the array $a[]$.
Only write accesses to $a[]$ have an influence on $\hash a$: a write $a[i]:=v$ is replaced by a sequence:
\begin{equation}\label{formula:array_multiset_write_sequence}
\hash a(a[i]) := \hash a(a[i])-1;~ a[i]:=v;~ \hash a(v) := \hash a(v)+1
\end{equation}
(that is, in addition to the array write, the count of elements for the value that gets overwritten is decremented, and the count of elements for the new value is incremented).

This auxiliary map $\hash a$ can itself be abstracted using our approach!
Let us now see how to implement this in our abstract forward semantics expressed using Horn clauses.
We enrich our Galois connection (Def.~\ref{def:Galois1}) as follows:
\begin{definition}\label{def:Galois1count}
The \emph{concretization} of $\abstr{I} \subseteq \chi \times (\iota \times \beta) \times (\beta \times \NN)$ is
\begin{multline}
\concretization[\hash]{\abstr{I}} = \Big\{ (\vx,a) \mid
  \forall i\in\iota~ \forall v \in \beta~\\
  \big(\vx,(i,a[i]),
  (v,\card \{ j \in \iota \mid a[j] = v \})\big) \in \abstr{I} \Big\}
\end{multline}
where $\card X$ denotes the number of elements in the set~$X$.

The \emph{abstraction} of $I \subseteq \chi \times \arraytype{\iota}{\beta}$ is
\begin{multline}
\abstraction[\hash]{I} = \Big\{ \big(\vx,(i,a[i]),(v,\card \{ j \in \iota \mid a[j] = v \})\big)\\ ~\Big|~ x \in \chi, i \in \iota \Big\}
\end{multline}
\end{definition}

\begin{theorem}
$\alpha_{\hash}$ and $\gamma_{\hash}$ form a Galois connection
\begin{equation*}
 \parts{\chi \times \arraytype{\iota}{\beta}}
 \galois{\alpha_{\hash}}{\gamma_{\hash}}
  \parts{\chi \times (\iota \times \beta) \times (\beta \times \NN)}.
\end{equation*}
\end{theorem}

The Horn rules for array reads and for scalar operations are the same as those for our first abstraction, except that we carry over the extra two components identically.

\begin{definition}[Read statement]\label{def:read1_count}
With the same notations in Def.~\ref{def:read1}:
{\small\begin{align}
\begin{aligned}
\forall \vx \in \chi~ \forall i \in \iota~ \forall v \in \beta~
  \forall k \in \iota~ \forall a_k,z \in \beta~ \forall a_{\hash z} \in \NN\\
  k \neq i \land \abstr{I}_1\big((\vx,i),(k,a_k),(z,a_{\hash z})\big) \\ \land
                 \abstr{I}_1\big((\vx,i),(i,v),(z,a_{\hash z})\big) \implies
  \abstr{I}_2(\vx,v,i,k,a_k)
\end{aligned}\rulespacing
\begin{aligned}
\forall \vx \in \chi~ \forall i \in \iota~ \forall v \in \beta \\
  \abstr{I}_1\big((\vx,i),(i,v),(z,a_{\hash z})\big) \implies
  \abstr{I}_2\big((\vx,v,i),(i,v),(z,a_{\hash z})\big)
\end{aligned}
\end{align}}
\end{definition}

\begin{lemma}
The abstract forward semantics of the read statement (Def.~\ref{def:read1_count}) is a sound abstraction of the concrete semantics given in Def.~\ref{def:read1}.
\end{lemma}

The abstraction of the write statement is more complicated (see the sequence of instructions in Formula~\ref{formula:array_multiset_write_sequence}). To move by a write operation $a[i]:=v$ from a control point $p_1$ to a control point $p_2$, we need two intermediate control points $p_a$ and $p_b$.

\begin{definition}[Write statement]\label{def:write1_count}
With the same notations in Def.~\ref{def:write1}:

{\small\begin{align*}
\begin{aligned}
\forall \vx \in \chi~ \forall i,k \in \iota~ \forall a_i,a_k,v,z \in \beta~
\forall a_{\hash z}\in \NN~ a_i \neq z  \land \\
\abstr{I}_1\big((\vx, v, i), (k, a_k), (z, a_{\hash z})\big) \land
\abstr{I}_1\big((\vx, v, i), (i, a_i), (z, a_{\hash z})\big) \\ \implies
\abstr{I}_a\big((\vx, v, i), (k, a_k), (z, a_{\hash z})\big)
\end{aligned}\rulespacing
\begin{aligned}
\forall \vx \in \chi~ \forall i,k \in \iota~ \forall a_i,a_k,v \in \beta~
\forall a_{\hash z}\in \NN \\
\abstr{I}_1\big((\vx, v, i), (k, a_k), (a_i, a_{\hash z})\big) \land
\abstr{I}_1\big((\vx, v, i), (i, a_i), (a_i, a_{\hash z})\big) \\ \implies
\abstr{I}_a\big((\vx, v, i), (k, a_k), (a_i, a_{\hash z}-1)\big)
\end{aligned}\rulespacing
\begin{aligned}
\forall \vx \in \chi~ \forall i,k \in \iota~ \forall a_i,a_k,v,z \in \beta~
\forall a_{\hash z}\in \NN \\ v \neq z \land
\abstr{I}_a\big((\vx, v, i), (k, a_k), (z, a_{\hash z})\big) \land
\abstr{I}_a\big((\vx, v, i), (i, a_i), (z, a_{\hash z})\big) \\ \implies
\abstr{I}_b\big((\vx, v, i), (k, a_k), (z, a_{\hash z})\big)
\end{aligned}\rulespacing
\begin{aligned}
\forall \vx \in \chi~ \forall i,k \in \iota~ \forall a_i,a_k,v \in \beta~
\forall a_{\hash z}\in \NN \\
\abstr{I}_a\big((\vx, v, i), (k, a_k), (v, a_{\hash z})\big) \land
\abstr{I}_a\big((\vx, v, i), (i, a_i), (v, a_{\hash z})\big) \\ \implies
\abstr{I}_b\big((\vx, v, i), (k, a_k), (v, a_{\hash z}+1)\big)
\end{aligned}\rulespacing
\begin{aligned}
\forall \vx \in \chi~ \forall i \in \iota~ \forall v \in \beta~
  \forall x \in \iota~ \forall a_k \in \beta~~  i\neq k \land \\
  \abstr{I}_1\big((\vx,i,v),(k,a_k), (z, a_{\hash z})\big)
  \implies
  \abstr{I}_2\big((\vx,v,i),(k,a_k), (z, a_{\hash z})\big)
\end{aligned}\rulespacing
\begin{aligned}
\forall \vx \in \chi~ \forall i \in \iota~ \forall v \in \beta~
  \forall k \in \iota~ \forall a_k \in \beta\\
  \abstr{I}_1\big((\vx,i,v),(i,a_k), (z, a_{\hash z})\big)  \implies
  \abstr{I}_2\big((\vx,v,i),(i,v), (z, a_{\hash z})\big)
\end{aligned}
\end{align*}}
\end{definition}

\begin{lemma}
The abstract forward semantics of the write statement (Def.~\ref{def:write1_count}) is a sound abstraction of the concrete semantics given in Def.~\ref{def:write1}.
\end{lemma}

If we want to compare the multiset of the contents of an array $a$ at the end of a procedure to its contents at the beginning of the procedure, one needs to keep a copy of the old multiset.
It is common that the property sought is a relation between the number of occurrences $\hash a(z)$ of an element $z$ in the output array $a$ and its number of occurrences $\hash a_0(z)$ in the input array $a^0$. In the above formulas, one may therefore replace the pair $(z, a_{\hash z})$ by $(z, a_{\hash z}, a^0_{\hash z})$, with $a^0_{\hash z}$ always propagated identically.

\begin{example}\label{ex:selection_sort_multiset}
Consider again selection sort (Program~\ref{prog:selection_sort}). We use the abstract semantics for read (Def.~\ref{def:read1_count}) and write (Def.~\ref{def:write1_count}), with an additional component $a^0_{\hash z}$ for tracking the original number of values $z$ in the array~$a$ (Listing~\ref{Horn:selection_sort_multiset}).

We specify the final property as the query
\begin{align}
\begin{aligned}
\forall l_0,h,k,a_k,z,a_{\hash z},a^0_{\hash z}~
\mathit{exit}(l_0,h,k,a_k,z,a_{\hash z},a^0_{\hash z}) \\
\implies a_{\hash z} = a^0_{\hash z}
\end{aligned}
\end{align}
\end{example}

\section{Counterexamples}
\label{sec:counterexample}
Solvers for Horn clauses based on counterexample-based abstraction refinement (CEGAR) construct a sequence of increasingly more precise abstractions of the Horn clause problem. At every step, they search for a \emph{counterexample} to the satisfiability of the Horn clauses: that is, a tree unfolding of the Horn clauses and matching assignments to the variables in the clauses, rooted at a violated query. If such a counterexample is found, the solver answers ``unsatisfiable'': this counterexample is a witness to the absence of solution of the system. If it is found not to exist, the solver examines its proof of nonexistence for clues how to refine the abstraction, typically by generating \emph{tree interpolants}, and the process goes on~\cite{DBLP:conf/cav/RummerHK13}.

In our case, a counterexample provided by the Horn solver proves the nonexistence of an inductive invariant capable of proving the desired properties \emph{in our abstraction}: it means that either our abstraction is too coarse, either the desired safety property is wrong because there exists a concrete counterexample.

\begin{figure}
{\small%
$
\infer[{v_2:=a[2]} \text{~(\ref{rule:counterexample_r4a})}]{\mathit{end}(0,1,1,0)}{
  \infer{r_4(0,1,0)}{
    \infer[{v1:=a[1]} \text{~(\ref{rule:counterexample_r3b})}]{r_3(1,0)}{
      \infer[{v_2:=a[2]} \text{~(\ref{rule:counterexample_r2a})}]{\mathit{cmp}(0,1,1,0)}{
        \infer[{v1:=a[1]} \text{~(\ref{rule:counterexample_r1b})}]{r_2(0,1,0)}{
          \infer{r_1(1,0)}{}
        } &
        \infer*{r_2(0,2,0)}{}
      }
    }
  } &
  \infer*{r_4(0,2,0)}{}
}
$
}
\caption{Counterexample unfolding of Ex~\ref{ex:counterexample} leading to $\mathit{end}(v_1.v_2,1,a_1)$ with $v_1 \leq v_2$, violating the condition}
\label{fig:counterexample_unfolding}
\end{figure}

\begin{example}\label{ex:counterexample}
The assertion at the end of this program is obviously valid (and may be established by, e.g., global value numbering):
\begin{lstlisting}
/* r1 */ v1 = a[1]; /* r2 */ v2 = a[2];
/* cmp */ assume(v1 == v2); /* kill(v1, v2) */
/* r3 */ v1 = a[1]; /* r4 */ v2 = a[2];
/* end */ assert(v1 == v2);
\end{lstlisting}

The system of Horn rules produced by Algorithm~\ref{algo:construct_Horn} is:
{\small\begin{align}
\forall k,a_k \in \ZZ~ r_1(k, a_k) \rulespacing
\begin{aligned}
\forall v_1,k,a_k \in \ZZ~
  r_1(k,a_k) \land r_1(1,v_1) \land k\neq 1\\
  \implies r_2(v_1, k, a_k)
\end{aligned}\rulespacing
\forall v_1~ r_1(1,v_1) \implies r_2(v_1,1,v_2) \label{rule:counterexample_r1b}\\
\begin{aligned}
\forall v_1,v_2,k,a_k \in \ZZ~
  r_2(v_1,k,a_k) \land r_2(v_1,2,v_2) \land k\neq 2\\ 
  \implies \mathit{cmp}(v_1,v_2,k,a_k)
\end{aligned} \label{rule:counterexample_r2a}\rulespacing
\forall v_1,v_2 \in \ZZ~ r_2(v_1,2,v_2) \implies \mathit{cmp}(v_1,v_2,2,v_2) \rulespacing
\forall v,k,a_k \in \ZZ~
  \mathit{cmp}(v,v,k,a_k) \implies r_3(k,a_k)  \label{rule:counterexample_cmp} \rulespacing
\begin{aligned}
\forall v_1,k,a_k \in \ZZ~
  r_3(k,a_k) \land r_3(1,v_1) \land k\neq 1\\
 \implies  r_4(v_1, k, a_k)
\end{aligned} \rulespacing
\forall v_1~ r_3(1,v_1) \implies r_4(v_1,1,v_1)  \label{rule:counterexample_r3b} \rulespacing
\begin{aligned}
\forall v_1,v_2,k,a_k \in \ZZ~
  r_4(v_1,k,a_k) \land r_4(v_1,2,v_2) \land k\neq 2\\
  \implies \mathit{end}(v_1,v_2,k,a_k)
\end{aligned}  \label{rule:counterexample_r4a} \rulespacing
\forall v_1,v_2 \in \ZZ~ r_4(v_1,2,v_2) \implies \mathit{end}(v_1,v_2,2,v_2) \rulespacing
\forall v_1,v_2,a_1 \in \ZZ~ \mathit{end}(v_1,v_2,1,a_1) \implies v_1=v_2
\end{align}}

This system is too abstract to prove the desired property; an abstract counterexample exists (Fig.~\ref{fig:counterexample_unfolding}).
\end{example}

Note that all our Horn rules are of the form%
\footnote{In this explanation we use a single $(k,a_k)$, as in Sec.~\ref{sec:abstraction1}, but the same carries to the use of multiple indices $(k_1,a_{k_1},k_2,a_{k_2})$ etc.}
\begin{equation}
\forall \dots~
  \abstr{I}_1\big(\dots, (k,a_k)\big) \land \dots \land \abstr{I}_1(\dots,\dots) \implies
  \abstr{I}_2\big(\dots, (k,a_k)\big) 
\end{equation}
thus, in the unfolding, the children of a node associated with a control point $p_2$ are all associated to the same control point $p_1$.
Furthermore, the rule associated to a node corresponds to one statement transitioning from $p_1$ to $p_2$.
Any branch from a leaf to the root of the unfolding thus corresponds to a sequence of statements from the original program.
It is, in the CEGAR view, an ``abstract counterexample trace'' from an initialization control point to a possible violation.
By conjoining the concrete semantics associated to each step we obtain a first-order formula over arithmetic and arrays.
If this formula is satisfiable, we have a concrete counterexample.

\begin{example*}[Ex.~\ref{ex:counterexample}, continued]
From the unfolding in Figure~\ref{fig:counterexample_unfolding}, one obtains the sequence of instructions to be tested for a concrete counterexample.
(On this example with no tests or loops, there is only one sequence from the start to the end of the program, but in general this says which test branches are taken.)
\end{example*}

\begin{algo}\label{algo:counterexample}
\begin{enumerate}
\item Construct the Horn clause system using Algorithm~\ref{algo:construct_Horn}.
\item Run the Horn clause solver. It returns ``satisfiable'', report ``proved''.
\item If it returns a counterexample unfolding, select a branch, collect the corresponding concrete transition relations and construct a trace satisfiability problem in first-order arithmetic plus arrays.
\item Run a satisfiability modulo theory (SMT) solver on this problem. If it returns ``satisfiable'', report ``violated''.
\item (Optional refinement step) Examine the array axioms in use in the unsatisfiability proof provided by the SMT-solver; increase the precision of the abstraction of these arrays by increasing the number of indices (e.g. move from a single index $k$ to two indices $k_1 \leq k_2$), and go back to step~1.
\end{enumerate}

Any branch in the unfolding could work, but we propose selecting the leftmost one according to the order in which we listed the antecedents of the Horn rules in this article.
\end{algo}

The values for the variables in the counterexample unfolding provided by the Horn clause solver may be used as hints for finding the values of the variables in the SMT problem.

\section{Related work}
\label{sec:related}

\subsection{Abstract interpretation}\label{sec:related_absint}
\paragraph{Smashing}
The simplest abstraction for an array is to ``smash'' all cells into a single one ---
this amounts to removing the $k$ component from our first Galois connection (Def.~\ref{def:Galois1}).
The weakness of that approach is that all writes are treated as ``may writes'': $a[i]:=x$ \emph{adds} the value $x$ to the set of values admissible for the array $a$, but there is no way to \emph{remove} any value from that set.
Such an approach thus cannot treat initialization loops (e.g. Program~\ref{prog:array_fill1}) precisely: it cannot prove that the old values have been erased.

\paragraph{Exploding}
At the other extreme, for an array of statically known finite length $N$ (which is common in embedded safety-critical software), one can distinguish all cells $a[0], \dots, a[N-1]$ and treat them as separate variables $a_0, \dots, a_{N-1}$;
e.g. a read \lstinline|x = a[i];| is treated as a
{\small
\begin{lstlisting}[mathescape=true]
switch (i) {
  case 0: x = a$_0$; break;
  $\dots$
  case $N-1$: x=a$_{N-1}$; break;
}
\end{lstlisting}}
This is a good solution when $N$ is small, but a terrible one when $N$ is large:
\begin{inparaenum}
\item many analysis approaches scale poorly with the number of active variables
\item an initialization loop will have to be unrolled $N$ times to show it initializes all cells.
\end{inparaenum}

Both these approaches have been used with success in the Astr\'ee static analyzer \cite{BlanchetCousotEtAl02-NJ,BlanchetCousotEtAl_PLDI03}, where the ``smashed'' cell or the individual array cells are typically further abstracted by intervals and other non-relational analyses.

\paragraph{Slices}
More sophisticated analyses \cite{GopanRS05,HalbwachsP08,peron:tel-00623697,perrelle:tel-00973892,CousotCL11} distinguish \emph{slices}
or \emph{segments} in the array, their boundaries depending on the index variables.
For instance, in array initialization (Program~\ref{prog:array_fill1}), such an analysis will tend to distinguish the area already initialized (indices $<i$) and the area yet to be initialized (indices $\geq i$).
In the simplest case, each slice is ``smashed'' into a single value, but more refined analyses express relationships between slices.
Since the slices are segments $[a,b]$ of indices, these analyses generalize poorly to multidimensional arrays.
Also, there is often a combinatorial explosion in analyzing how array slices may or may not overlap.

To our best knowledge, all these approaches factor through our Galois connections $\galois{\alpha}{\gamma}$, $\galois{\alpha_{2\leq}}{\gamma_{2\leq}}$ or combinations thereof: that is, their abstraction can be expressed as a composition of our abstraction and further abstraction
--- even though our implementation of the abstract transfer functions is completely different from theirs.
Our approach, however, separates the concerns of
\begin{inparaenum}[i)]
\item abstracting array problems to array-less problems
\item abstracting the relationships between different cells and indices.
\end{inparaenum}

\subsection{Array removal by program transformation}
\label{sec:abstraction_weak_transfo}
\citet{Monniaux_Alberti_SAS2015} recently published a method for analyzing array programs by transforming them into array-free programs.
They use the same Galois connections ($\galois{\alpha}{\gamma}$, $\galois{\alpha_{2\leq}}{\gamma_{2\leq}}$) as us, but they implement the abstract transfer functions differently.
While we transform the program into a system of non-linear Horn clauses, they transform it into another program without arrays.

\begin{definition}[Array analysis by transformation to an array-free program]\label{def:weak1_transfo}
Non-array operations are left unchanged. A read \lstinline|v=a[i];| is transformed into
\begin{lstlisting}
if (i==k) x=ak; else havoc(x);
\end{lstlisting} and a write \lstinline|a[i]=v;| is transformed into
\begin{lstlisting}
if (i==k) ak=x;
\end{lstlisting}
\lstinline|havoc(x)| sets \lstinline|x| to a nondeterministic value.
\end{definition}

Note that the ``flat'' encoding of programs into Horn clauses yields a system of linear clauses. Thus, chaining \citet{Monniaux_Alberti_SAS2015} and a tool for turning scalar program analysis problems into a system of Horn clauses would yield linear clauses: the resulting encoding would thus be different from the one produced by our approach.
The reason is that their abstraction is actually weaker than our abstraction. Let us see the Horn clauses corresponding to the encoding of the ``read'' operation in their approach.

\begin{definition}[Read statement, weakened]\label{def:read1_weakened}
With the same notations as in Def.~\ref{def:read1}, another forward abstract semantics for  \lstinline|v=a[i];| is given by the Horn clauses:
\begin{align}
\begin{aligned}
\forall \vx \in \chi~ \forall i \in \iota~ \forall v \in \alpha~
  \forall k \in \iota~ \forall a_k \in \alpha\\
  i \neq k \land \abstr{I}_1(\vx,i,k,a_k) \implies
  \abstr{I}_2(\vx,v,i,k,a_k)
\end{aligned}\rulespacing
\begin{aligned}
\forall \vx \in \chi~ \forall i \in \iota~ \forall v \in \alpha~
  \forall k \in \iota~ \forall a_k \in \alpha\\
  \abstr{I}_1(\vx,i,i,v) \implies
  \abstr{I}_2(\vx,v,i,i,v)
\end{aligned}
\end{align}
\end{definition}

\begin{lemma}
This forward abstract semantics is sound and equivalent to the forward semantics of the transformation of a read statement according to Def.~\ref{def:weak1_transfo}.
\end{lemma}

\begin{proof}
Soundness follows from this definition over-approximating Def.~
\ref{def:read1} (removal of one conjunct). Equivalence to the forward semantics of the transformed read statement is obvious.
\end{proof}

Using this weakened semantics to abstract arrays $a$ and $b$ by quadruplets $(x,a[x],y,b[y])$, the correctness of Ex.~\ref{ex:array_reverse_once1} cannot be proved.
\citet{Monniaux_Alberti_SAS2015} are able to prove the correctness of this program by using a more expensive abstraction, where array $a$ (current working array of the reversal) and array $b$ (original values of $a$) are abstracted using a set $\abstr{I}$ of sextuplets such that $\forall 0 \leq x \leq y < n~ \forall 0 \leq z < n ~ (x,a[x],y,a[y],z,b[z]) \in \abstr{I}$; that is, $a$ is abstracted as in Sec.~\ref{sec:sortedness}.

\citet[Sec.~5.5]{Monniaux_Alberti_SAS2015} are sometimes able to recover the loss of precision induced by their abstractions by applying a form of \emph{quantifier elimination}.
When they have an abstract state $(\vx, k_1, a_{k_1}, k_2, a_{k_2})$, they reason that a state is spurious if there exist $k'_1$ ($k'_1 \neq k_1$, $k'_1 \leq k_2$) such that there is no abstract state $(\vx, k'_1, a_{k'_1}, k_2, a_{k_2})$. In intuitive terms, this means that if $a[k_2] = a_{k_2}$ then there is no way to fill the value at cell $a[k_1]$ --- but since all cells in the array must have a value, this means that $a[k_2] = a_{k_2}$ is impossible.

Thus, one could have a ``filtering'' or \emph{reduction} rule
\begin{align}
\begin{aligned}
  \left(\forall k'_1~ \left(k'_1 \leq k_2 \implies \exists a_{k'_1} \abstr{I}_1(\vx, k'_1, a_{k'_1}, k_2, a_{k_2})\right) \right) \\ \land
  \abstr{I}_1(\vx, k_1, a_{k_1}, k_2, a_{k_2}) \implies
  \abstr{I}_2(\vx, k_1, a_{k_1}, k_2, a_{k_2})
\end{aligned}
\end{align}
which would not change the concretization --- $\concretization[2\leq]{\abstr{I}_2} = \concretization[2\leq]{\abstr{I}_1}$ --- but reduce the abstract state, which may later yield a more precise result (applying the same sound abstract operation to two sets of abstract states with the same concretization may yield two sets of abstract states with different concretizations).

We cannot specify such a \emph{reduction rule} using Horn clauses (because a $\forall$ on the left of $\implies$ is effectively an existential in the prenex form, which is banned).
However, we can easily specify a partial filtering by instantiating the universal quantifier on certain values, thereby obtaining a finite conjunction in the antecedent of the implication.
This is, in essence, what we gain by the use of non-linear Horn clauses.

Thus, using non-linear Horn clauses, we are able to integrate partial reductions in the abstract domain to the fixed-point problem to solve, whereas \citet[Sec.~5.5]{Monniaux_Alberti_SAS2015} had to first solve the full fixed point problem (analysis of the transformed program), then perform reductions.
In general, it is more precise to solve a fixed point problem using a precise operator $f$ than to solve it using an imprecise operator $g \geq f$, then reduce the final result.
We believe therefore that our approach improves in this respect upon that of \citeauthor{Monniaux_Alberti_SAS2015}'s.

Another difficulty they obviously faced was the limitations of the back-end solvers that they could use. The integer acceleration engine \soft{Flata} severely limits the kind of transition relations that can be considered and scales poorly.
The abstract interpreter \soft{ConcurInterproc} can infer disjunctive properties (necessary to distinguish two slices in an array) only if given case splits using observer Boolean variables; but the cost increases greatly (exponentially, in the worst case) with the number of such variables.

\subsection{Predicate abstraction, CEGAR and array interpolants}
There exist a variety of approaches based on counterexample-guided abstraction refinement using interpolants (see also Sec.~\ref{sec:counterexample}).
In a nutshell: let $\big(\tau_i(\vx_i,\vx_{i+1})\big)_{0 \leq i < n}$ be the transition relations associated to a sequence of statements (assignments and guards), where $\vx_i$ is the vector of active program variables after $i$~steps.
It is impossible to reach the end of the sequence from the beginning if and only if this formula is unsatisfiable:
\begin{equation}
\tau_0(\vx_0,\vx_1) \land \dots \land \tau_{n-1}(\vx_{n-1},\vx_n)
\end{equation}
The proof of unsatisfiability of this formula, as obtained from a satisfiability modulo theory (SMT) solver, may be convoluted. For the purpose of inferring useful ``candidate invariants'' on the program, we would prefer ``local'' arguments $I_i$, talking only about the variables at a given step:
\begin{equation}
I_i(\vx_i) \land \tau_0(\vx_i,\vx_{i+1}) \implies I_{i+1}(\vx_{i+1})
\end{equation}
and $I_n = \mathsf{false}$.
Such $I_i$ are known as \emph{Craig interpolants} \cite{DBLP:conf/apn/McMillan05,McMillan06,McMillan11} and are typically obtained by reprocessing the proof of unsatisfiability from the SMT solver.
One difficulty with that approach is that not all interpolants are equally interesting: one seeks interpolants that not only prove that an individual sequence of statements leading to a bad state is infeasible, but that generalize well and can be used in a proof that many sequences of statements leading to a bad state are infeasible, hopefully leading to a proof that no sequence can lead to a bad state.

Generating good interpolants from purely arithmetic problems is already a difficult problem, and generating good universally quantified interpolants on array properties has proved even more challenging \cite{JhalaM07,AlbertiBGRS14,Alberti_Monniaux_SAC-SVT2015}.

\subsection{Acceleration}
It is possible to compute exactly the transitive closure of some transition relations, and thus to summarize some loop exactly. The class of transition relations supported is however restricted.

\citet{BozgaHIKV09} have proposed a method for accelerating certain transition relations involving actions over arrays, which outputs the transitive closure in the form of a \emph{counter automaton}.
Translating the counter automaton into a first-order formula expressing the array properties however results in a loss of precision.


\section{Conclusion and perspectives}
\label{sec:conclusion}
We have proposed a generic approach to abstract programs and universal properties over arrays (and, more generally, arbitrary maps) by syntactic transformation into a system of Horn clauses without arrays, which is then sent to a solver.
This transformation is powerful enough that it can be used to prove, fully automatically and within minutes, that the output of selection sort is sorted and is a permutation of the input.

While some solvers have difficulties with the kind of Horn systems that we generate, some (e.g. \soft{Spacer}) are capable of solving them quite well.
We have used the stock version of the solvers, without help from their designers or special tuning, thus higher performance is to be expected in the future.
Indeed, we feel the kind of systems we generate would make good benchmarks for Horn solvers.
If the solver cannot find the invariants on its own, it can be helped by partial invariants from the user.
Also, if it finds a counterexample in the abstraction, we propose a method for reconstructing a concrete counterexample (Sec.~\ref{sec:counterexample}) or triggering a refinement.

\paragraph{Existentials}
Our approach can be used, \emph{a fortiori}, to prove or infer quantifier-free properties, but not existentials. Future work could include quantifier instantiation heuristics for existentials.

\paragraph{Backward analysis}
Our rules are for ``forward analysis'': they express that if configuration is possible at one step during one execution, then some configuration may be possible at the next step during that execution.
We thus define a super-set of all states reachable from program initialization, and the desired property is proved if this set is included in the property.

An alternative approach is ``backward analysis'': find a super-set of the set of all states reachable from a property violation, such that this set has empty intersection with the initial states.
A possible research direction would be to derive backward rules and compare their efficiency to that of forward rules.

\paragraph{Procedures}
One approach to procedures is to consider a call to a procedure as jump to the first node of the callee and a return as a jump back to each possible caller node.
Because this mixes together all calls to the same procedure, it can lose a lot of precision; some tracking variables, abstracting the stack (in the simplest case, the topmost call site), may be added to avoid precision loss.
Such an approach may be immediately combined with ours.

In contrast, some other approaches encode procedures (or other program fragments, such as the loop bodies) as binary input/output relations over the variable state --- or, rather, the fragment of the state that may be read or written by the procedure. This maps well to Horn clauses: in the solution of the solver, the predicate associated to a procedure summarizes its action.
How to combine this vision with our approach is a topic for future research.

\paragraph{High-level maps and sets}
Many programming languages provide libraries for finite maps and (multi)sets. In this article, we have explained how to abstract some, but not all of their features (Sec.~\ref{sec:simple_sets}) --- for instance we do not provide an iterator for non-integer set element types.
Future work should include reviewing their features and common usage in order to design suitable abstractions.

\paragraph{Query-less analysis} One advantage of some of earlier approaches (the abstract interpretation ones from Sec.~\ref{sec:related_absint} and the program transformation from \citet{Monniaux_Alberti_SAS2015}) is that they are capable of inferring what a program does, or at least a meaningful abstraction of it (e.g. ``at the end of this program all cells in the array $a$ contains $42$'')  as opposed to merely proving a property supplied by the user.
Our approach can achieve this as well, \emph{provided it is used with a Horn clause solver that does not require queries} and still provides some interesting solution (a query-less Horn problem has a trivial, uninteresting solution: ``true'' to all predicates).

This Horn clause solver should however be capable of generating disjunctive properties (e.g. $(k < i \land a_k = 0) \lor (k \geq i \land a_k = 42)$); thus a simple approach by abstract interpretation of the Horn clauses in, say, a sub-class of the convex polyhedra, will not do.
We know of no such Horn solver; building one is an interesting research challenge.
Maybe certain partitioning approaches used in sequential program verification \cite{DBLP:journals/toplas/RivalM07,HenryMM12} may be transposed to Horn clauses.

We have expressed an abstraction of the semantics of programs with array reads and writes into a system of Horn clauses on scalar variables.
Another approach would be to directly work from Horn clauses on array variables, and over-approximate the rules and under-approximate the queries into an array-free Horn problem.
\smallskip

\paragraph{Objects} We have considered simple programs operating over arrays or maps, as opposed to a real-life programming language with objects, references or, horror, pointer arithmetic. Yet, our approach can be adapted to such languages.
\lstset{language=Java}
One can indeed see each object field name in a language such as Java (e.g. \lstinline|String x;|) as a map from object references to values (here, of type \lstinline|String|).
The reference may be an index (perhaps $i$ if the object is the $i$-th object allocated) or a more complex record of the site of allocation.

\paragraph{Pointers} Languages with pointers, pointer arithmetic and, worse, access to an object of a type through a pointer of an incompatible type (not uncommon in traditional C programming), can be handled by seeing the memory as an array of bytes, but this leads to impractically inefficient analysis.
It is however often possible to segment the memory into independent variables (never accessed through pointers, or at least accessed only through pointers at known locations) and a number of disjoint arrays.
Our analysis can then be used over these arrays.

\printbibliography

\newpage
\appendix

\section{Horn clause problems}
\begin{lstlisting}[language=SMT,caption={Array fill 1D},label=Horn:array_fill1]
(set-logic HORN)

(declare-fun loop (Int Int Int Int) Bool) ; n i k a[k]
(declare-fun write (Int Int Int Int) Bool) ; n i k a[k]
(declare-fun incr (Int Int Int Int) Bool) ; n i k a[k]
(declare-fun end (Int Int Int) Bool) ; n a[]

(assert (forall ((n Int) (k Int) (ak Int))
  (=> (and (<= 0 k) (< k n)) (loop n 0 k ak))))

(assert (forall ((n Int) (i Int) (k Int) (ak Int))
  (=> (and (< i n) (loop n i k ak)) (write n i k ak))))

(assert (forall ((n Int) (i Int) (k Int) (ak Int))
  (=> (and (distinct i k) (write n i k ak)) (incr n i k ak))))

(assert (forall ((n Int) (i Int) (k Int) (ak Int))
  (=> (write n i k ak) (incr n i i 42))))

(assert (forall ((n Int) (i Int) (k Int) (ak Int))
  (=> (incr n i k ak) (loop n (+ i 1) k ak))))

(assert (forall ((n Int) (i Int) (k Int) (ak Int))
  (=> (and (>= i n) (loop n i k ak)) (end n k ak))))

(assert (forall ((n Int) (k Int) (ak Int))
  (=> (and (>= k 0) (< k n) (end n k ak)) (= ak 42))))

(check-sat)
(get-model)
\end{lstlisting}

\begin{lstlisting}[language=SMT,caption={Array fill 1D, even-odd},label=Horn:array_fill1_even_odd]
(set-logic HORN)

(declare-fun loop (Int Int Int Int) Bool) ; n i k a[k]
(declare-fun write (Int Int Int Int) Bool) ; n i k a[k]
(declare-fun incr (Int Int Int Int) Bool) ; n i k a[k]
(declare-fun end (Int Int Int) Bool) ; n a[]

(assert (forall ((n Int) (k Int) (ak Int))
  (=> (and (<= 0 k) (< k n)) (loop n 0 k ak))))

(assert (forall ((n Int) (i Int) (k Int) (ak Int))
  (=> (and (< i n) (loop n i k ak)) (write n i k ak))))

(assert (forall ((n Int) (i Int) (k Int) (ak Int))
  (=> (and (distinct i k) (write n i k ak)) (incr n i k ak))))

(assert (forall ((n Int) (i Int) (k Int) (ak Int))
  (=> (write n i k ak) (incr n i i (mod i 2)))))

(assert (forall ((n Int) (i Int) (k Int) (ak Int))
  (=> (incr n i k ak) (loop n (+ i 1) k ak))))

(assert (forall ((n Int) (i Int) (k Int) (ak Int))
  (=> (and (>= i n) (loop n i k ak)) (end n k ak))))

(assert (forall ((n Int) (k Int) (ak Int))
  (=> (end n (* 2 k) ak) (= ak 0))))

(assert (forall ((n Int) (k Int) (ak Int))
  (=> (end n (+ (* 2 k) 1) ak) (= ak 1))))

(check-sat)
\end{lstlisting}

\begin{lstlisting}[language=SMT,caption={Array reverse},label=Horn:array_reverse_once1]
(set-logic HORN)

(declare-fun init (Int Int Int Int Int) Bool) ; n k a[k] l a0[l]
(declare-fun loop (Int Int Int Int Int Int) Bool) ; n i k a[k] l a0[l]
(declare-fun read1 (Int Int Int Int Int Int) Bool) ; n i k a[k] l a0[l]
(declare-fun read2 (Int Int Int Int Int Int Int) Bool) ; n i tmp1 k a[k] l a0[l]
(declare-fun write1 (Int Int Int Int Int Int Int Int) Bool) ; n i tmp1 tmp2 k a[k] l a0[l]
(declare-fun write2 (Int Int Int Int Int Int Int) Bool) ; n i tmp1 k a[k] l a0[l]
(declare-fun incr (Int Int Int Int Int Int) Bool) ; n i k a[k] l a0[l]
(declare-fun end (Int Int Int Int Int) Bool) ; n k a[k] l a0[l]

(assert (forall ((n Int) (k Int) (ak Int))
  (=> (and (<= 0 k) (< k n))
      (init n k ak k ak))))

(assert (forall ((n Int) (k Int) (ak Int) (l Int) (a0l Int))
  (=> (and (<= 0 k) (< k n) (<= 0 l) (< l n) (distinct k l))
      (init n k ak l a0l))))

(assert (forall ((n Int) (k Int) (ak Int) (l Int) (a0l Int))
  (=> (init n k ak l a0l)
      (loop n 0 k ak l a0l))))

(assert (forall ((n Int) (i Int) (k Int) (ak Int) (l Int) (a0l Int))
  (let ((j (- n (+ i 1))))
    (=> (and (< i j) (loop n i k ak l a0l))
        (read1 n i k ak l a0l)))))

(assert (forall ((n Int) (i Int) (tmp1 Int)
    (k Int) (ak Int) (l Int) (a0l Int))
  (let ((j (- n (+ i 1))))
    (=> (and (distinct i k)
             (read1 n i k ak l a0l)
             (read1 n i i tmp1 l a0l))
        (read2 n i tmp1 k ak l a0l)))))

(assert (forall ((n Int) (i Int) (tmp1 Int)
    (k Int) (ak Int) (l Int) (a0l Int))
  (let ((j (- n (+ i 1))))
    (=> (read1 n i i tmp1 l a0l)
        (read2 n i tmp1 i tmp1 l a0l)))))

(assert (forall ((n Int) (i Int) (tmp1 Int) (tmp2 Int)
    (k Int) (ak Int) (l Int) (a0l Int))
  (let ((j (- n (+ i 1))))
    (=> (and (distinct j k)
             (read2 n i tmp1 k ak l a0l)
             (read2 n i tmp1 j tmp2 l a0l))
        (write1 n i tmp1 tmp2 k ak l a0l)))))

(assert (forall ((n Int) (i Int) (tmp1 Int) (tmp2 Int)
    (k Int) (ak Int) (l Int) (a0l Int))
  (let ((j (- n (+ i 1))))
    (=> (and (read2 n i tmp1 j tmp2 l a0l))
        (write1 n i tmp1 tmp2 j tmp2 l a0l)))))

(assert (forall ((n Int) (i Int) (tmp1 Int) (tmp2 Int)
    (k Int) (ak Int) (l Int) (a0l Int))
  (let ((j (- n (+ i 1))))
    (=> (and (write1 n i tmp1 tmp2 k ak l a0l) (distinct i k))
        (write2 n i tmp1 k ak l a0l)))))

(assert (forall ((n Int) (i Int) (tmp1 Int) (tmp2 Int)
    (ak Int) (l Int) (a0l Int))
  (let ((j (- n (+ i 1))))
    (=> (write1 n i tmp1 tmp2 i ak l a0l)
        (write2 n i tmp1 i tmp2 l a0l)))))

(assert (forall ((n Int) (i Int) (tmp1 Int)
                 (k Int) (ak Int) (l Int) (a0l Int))
  (let ((j (- n (+ i 1))))
    (=> (and (write2 n i tmp1 k ak l a0l) (distinct j k))
        (incr n i k ak l a0l)))))

(assert (forall ((n Int) (i Int) (tmp1 Int)
                 (ak Int) (l Int) (a0l Int))
  (let ((j (- n (+ i 1))))
    (=> (write2 n i tmp1 j ak l a0l)
        (incr n i j tmp1 l a0l)))))

(assert (forall ((n Int) (i Int)
                 (k Int) (ak Int) (l Int) (a0l Int))
  (let ((j (- n (+ i 1))))
    (=> (incr n i k ak l a0l)
        (loop n (+ i 1) k ak l a0l)))))

(assert (forall ((n Int) (i Int)
                 (k Int) (ak Int) (l Int) (a0l Int))
  (let ((j (- n (+ i 1))))
    (=> (and (>= i j) (loop n i k ak l a0l))
        (end n k ak l a0l)))))

(assert (forall ((n Int) (i Int) (ak Int) (a0l Int))
  (let ((j (- n (+ i 1))))
    (=> (and (>= i 0) (< i n)
	     (end n i ak j a0l))
	(= ak a0l)))))

;(assert (forall ((n Int) (i Int) (ai Int) (j Int) (a0j Int))
;  (not (end n i ai j a0j))))
 
(check-sat)
\end{lstlisting}

\begin{lstlisting}[language=SMT,caption={Real-index maps},label=Horn:real_indexed_maps]
(set-logic HORN)

(declare-fun init (Real Int) Bool) ; x ax
(declare-fun w1 (Real Int) Bool) ; x ax
(declare-fun w2 (Real Int) Bool) ; x ax
(declare-fun exit (Real Int) Bool) ; x ax

(assert (forall ((x Real)) (init x 0)))

(assert (forall ((x Real) (ax Int))
  (=> (and (init x ax) (distinct x 1))
      (w1 x ax))))

(assert (forall ((ax Int))
  (=> (init 1 ax) (init 1 10))))

(assert (forall ((x Real) (ax Int))
  (=> (and (w1 x ax) (distinct x 2))
      (w2 x ax))))

(assert (forall ((ax Int))
  (=> (w1 2 ax) (w2 2 20))))

(assert (forall ((x Real) (ax Int))
  (=> (and (w2 x ax) (distinct x 3))
      (exit x ax))))

(assert (forall ((ax Int))
  (=> (w2 3 ax) (exit 3 30))))

(assert (forall ((x Real) (ax Int))
  (=> (exit x ax) (>= ax 0))))

(assert (forall ((x Real) (ax Int))
  (=> (exit x ax) (<= ax 30))))

(check-sat)
\end{lstlisting}

\begin{lstlisting}[language=SMT,caption={Array fill 2D},label=Horn:array_fill2]
(set-logic HORN)

(declare-fun init (Int Int Int Int Int) Bool) ; m n x y a[x,y]
(declare-fun loopi (Int Int Int Int Int Int) Bool) ; m n i x y a[x,y]
(declare-fun loopj (Int Int Int Int Int Int Int) Bool) ; m n i j x y a[x,y]
(declare-fun write (Int Int Int Int Int Int Int) Bool) ; m n i j x y a[x,y]
(declare-fun incrj (Int Int Int Int Int Int Int) Bool) ; m n i j x y a[x,y]
(declare-fun incri (Int Int Int Int Int Int) Bool) ; m n i x y a[x,y]
(declare-fun end (Int Int Int Int Int) Bool) ; m n x y a[x,y]

(assert (forall ((m Int) (n Int) (x Int) (y Int) (axy Int))
  (=> (and (<= 0 x) (< x m) (<= 0 y) (< y n)) (init m n x y axy))))

(assert (forall ((m Int) (n Int) (x Int) (y Int) (axy Int))
  (=> (init m n x y axy) (loopi m n 0 x y axy))))

(assert (forall ((m Int) (n Int) (i Int) (x Int) (y Int) (axy Int))
  (=> (and (< i m) (loopi m n i x y axy))
      (loopj m n i 0 x y axy))))

(assert (forall ((m Int) (n Int) (i Int) (x Int) (y Int) (axy Int))
  (=> (and (>= i m) (loopi m n i x y axy))
      (end m n x y axy))))

(assert (forall ((m Int) (n Int) (i Int) (j Int) (x Int) (y Int) (axy Int))
  (=> (and (< j n) (loopj m n i j x y axy))
      (write m n i j x y axy))))

(assert (forall ((m Int) (n Int) (i Int) (j Int) (x Int) (y Int) (axy Int))
  (=> (and (>= j n) (loopj m n i j x y axy))
      (incri m n i x y axy))))

(assert (forall ((m Int) (n Int) (i Int) (j Int) (x Int) (y Int) (axy Int))
  (=> (and (write m n i j x y axy)
           (not (and (= i x) (= j y))))
      (incrj m n i j x y axy))))

(assert (forall ((m Int) (n Int) (i Int) (j Int) (aij Int))
  (=> (write m n i j i j aij)
      (incrj m n i j i j 42))))

(assert (forall ((m Int) (n Int) (i Int) (j Int) (x Int) (y Int) (axy Int))
  (=> (incrj m n i j x y axy)
      (loopj m n i (+ j 1) x y axy))))

(assert (forall ((m Int) (n Int) (i Int) (x Int) (y Int) (axy Int))
  (=> (incri m n i x y axy)
      (loopi m n (+ i 1) x y axy))))

(assert (forall ((m Int) (n Int) (x Int) (y Int) (axy Int))
  (=> (end m n x y axy) (= axy 42))))

(check-sat)
\end{lstlisting}

\begin{lstlisting}[language=SMT,caption={Find minimum},label=Horn:find_minimum]
(set-logic HORN)

(declare-fun init (Int Int Int Int) Bool) ; l h k a[k]
(declare-fun read1 (Int Int Int Int) Bool) ; l h k a[k]

(declare-fun loop (Int Int Int Int Int Int Int) Bool) ; l h i p b k a[k]

(declare-fun read2 (Int Int Int Int Int Int Int) Bool) ; l h i p b k a[k]
(declare-fun test (Int Int Int Int Int Int Int Int) Bool) ; l h i p b v k a[k]

(declare-fun end (Int Int Int Int Int Int) Bool) ; l h p b k a[k]

(assert (forall ((l Int) (h Int) (k Int) (ak Int) (b Int))
  (=> (and (init l h k ak) (< l (- h 1)))
      (read1 l h k ak))))

(assert (forall ((l Int) (h Int) (k Int) (ak Int) (b Int))
  (=> (and (read1 l h k ak)
           (read1 l h l b)
           (distinct k l))
      (loop l h (+ l 1) l b k ak))))

(assert (forall ((l Int) (h Int) (b Int))
  (=> (read1 l h l b)
      (loop l h (+ l 1) l b l b))))

(assert (forall ((l Int) (h Int) (i Int) (p Int) (b Int) (k Int) (ak Int))
  (=> (and (loop l h i p b k ak) (< i h))
      (read2 l h i p b k ak))))

(assert (forall ((l Int) (h Int) (i Int) (p Int) (b Int) (v Int) (k Int) (ak Int))
  (=> (and (read2 l h i p b k ak)
           (read2 l h i p b i v)
           (distinct k i))
      (test l h i p b v k ak))))

(assert (forall ((l Int) (h Int) (i Int) (p Int) (b Int) (v Int))
  (=> (read2 l h i p b i v)
      (test l h i p b v i v))))

(assert (forall ((l Int) (h Int) (i Int) (p Int) (b Int) (v Int) (k Int) (ak Int))
  (=> (and (test l h i p b v k ak)
           (< v b))
      (loop l h (+ i 1) i v k ak))))

(assert (forall ((l Int) (h Int) (i Int) (p Int) (b Int) (v Int) (k Int) (ak Int))
  (=> (and (test l h i p b v k ak)
           (>= v b))
      (loop l h (+ i 1) p b k ak))))

(assert (forall ((l Int) (h Int) (i Int) (p Int) (b Int) (k Int) (ak Int))
  (=> (and (loop l h i p b k ak)
           (>= i h))
      (end l h p b k ak))))

; Initialization
(assert (forall ((l Int) (h Int) (k Int) (ak Int))
  (init l h k ak)))

; Properties to prove
(assert (forall ((l Int) (h Int) (p Int) (b Int) (k Int) (ak Int))
  (=> (end l h p b k ak)
      (>= p l))))

(assert (forall ((l Int) (h Int) (p Int) (b Int) (k Int) (ak Int))
  (=> (end l h p b k ak)
      (< p h))))

(assert (forall ((l Int) (h Int) (p Int) (b Int) (k Int) (ap Int))
  (=> (end l h p b p ap)
      (= b ap))))

(assert (forall ((l Int) (h Int) (p Int) (b Int) (k Int) (ap Int) (ak Int))
  (=> (and (>= k l)
           (< k h)
           (end l h p b k ak))
      (<= b ak))))

(check-sat)
\end{lstlisting}

\begin{lstlisting}[language=SMT,caption={Selection sort: sortedness},label=Horn:selection_sort_sortedness]
; SELECTION SORT
; Solved in 6' by Spacer 2015-07-01_0dfed9c654a6f715b8660ff1d27de74fbfc41916

(set-logic HORN)

(declare-fun init (Int Int Int Int Int Int) Bool) ; l0 h k a[k] k2 a[k2]
(declare-fun outerloop (Int Int Int Int Int Int Int) Bool) ; l0 l h k a[k] k2 a[k2]
(declare-fun exit (Int Int Int Int Int Int) Bool) ; l0 h k a[k] k2 a[k2]

(declare-fun read1 (Int Int Int Int Int Int Int) Bool) ; l0 l h k a[k] k2 a[k2]

(declare-fun loop (Int Int Int Int Int Int Int Int Int Int Int) Bool) ; l0 l h i p b f k a[k] k2 a[k2]

(declare-fun read2 (Int Int Int Int Int Int Int Int Int Int Int) Bool) ; l0 l h i p b f k a[k] k2 a[k2]
(declare-fun test (Int Int Int Int Int Int Int Int Int Int Int Int) Bool) ; l0 l h i p b v f k a[k] k2 a[k2]

(declare-fun write1 (Int Int Int Int Int Int Int Int Int Int) Bool) ; l0 l h p b f k a[k] k2 a[k2]
(declare-fun write2 (Int Int Int Int Int Int Int Int Int Int) Bool) ; l0 l h p b f k a[k] k2 a[k2]
(declare-fun endswap (Int Int Int Int Int Int Int Int Int Int) Bool) ; l0 l h p b f k a[k] k2 a[k2]
(declare-fun incr (Int Int Int Int Int Int Int) Bool) ; l0 l h k a[k] k2 a[k2]

(assert (forall ((l0 Int) (h Int) (k Int) (ak Int) (k2 Int) (ak2 Int))
  (=> (and (init l0 h k ak k2 ak2) (<= l0 h))
      (outerloop l0 l0 h k ak k2 ak2))))

(assert (forall ((l0 Int) (l Int) (h Int) (k Int) (ak Int) (k2 Int) (ak2 Int))
  (=> (and (outerloop l0 l h k ak k2 ak2) (< l (- h 1)))
      (read1 l0 l h k ak k2 ak2))))

(assert (forall ((l0 Int) (l Int) (h Int) (k Int) (ak Int) (k2 Int) (ak2 Int))
  (=> (and (outerloop l0 l h k ak k2 ak2) (>= l (- h 1)))
      (exit l0 h k ak k2 ak2))))

(assert (forall ((l0 Int) (l Int) (h Int) (b Int) (k Int) (ak Int) (k2 Int) (ak2 Int))
  (=> (and (read1 l0 l h k ak k2 ak2)
           (read1 l0 l h l b k2 ak2)
           (distinct k l) (< l k2))
      (loop l0 l h (+ l 1) l b b k ak k2 ak2))))

(assert (forall ((l0 Int) (l Int) (h Int) (b Int) (k Int) (ak Int) (k2 Int) (ak2 Int))
  (=> (and (read1 l0 l h k ak k2 ak2)
           (read1 l0 l h k ak l b)
           (distinct k2 l) (< k l))
      (loop l0 l h (+ l 1) l b b k ak k2 ak2))))

(assert (forall ((l0 Int) (l Int) (h Int) (b Int) (k2 Int) (ak2 Int))
  (=> (and (< l k2) (read1 l0 l h l b k2 ak2))
      (loop l0 l h (+ l 1) l b b l b k2 ak2))))

(assert (forall ((l0 Int) (l Int) (h Int) (b Int) (k Int) (ak Int))
  (=> (and (<= k l) (read1 l0 l h k ak l b))
      (loop l0 l h (+ l 1) l b b k ak l b))))

(assert (forall ((l0 Int) (l Int) (h Int) (i Int) (p Int) (b Int) (f Int) (k Int) (ak Int) (k2 Int) (ak2 Int))
  (=> (and (loop l0 l h i p b f k ak k2 ak2) (< i h))
      (read2 l0 l h i p b f k ak k2 ak2))))

(assert (forall ((l0 Int) (l Int) (h Int) (i Int) (p Int) (b Int) (v Int) (f Int) (k Int) (ak Int) (k2 Int) (ak2 Int))
  (=> (and (read2 l0 l h i p b f k ak k2 ak2)
           (read2 l0 l h i p b f i v k2 ak2)
           (distinct k i) (< i k2))
      (test l0 l h i p b v f k ak k2 ak2))))

(assert (forall ((l0 Int) (l Int) (h Int) (i Int) (p Int) (b Int) (v Int) (f Int) (k Int) (ak Int) (k2 Int) (ak2 Int))
  (=> (and (read2 l0 l h i p b f k ak k2 ak2)
           (read2 l0 l h i p b f k ak i v)
           (distinct k2 i) (< k i))
      (test l0 l h i p b v f k ak k2 ak2))))

(assert (forall ((l0 Int) (l Int) (h Int) (i Int) (p Int) (b Int) (v Int) (f Int) (k2 Int) (ak2 Int))
  (=> (and (< i k2) (read2 l0 l h i p b f i v k2 ak2))
      (test l0 l h i p b v f i v k2 ak2))))

(assert (forall ((l0 Int) (l Int) (h Int) (i Int) (p Int) (b Int) (v Int) (f Int) (k Int) (ak Int))
  (=> (and (<= k i) (read2 l0 l h i p b f k ak i v))
      (test l0 l h i p b v f k ak i v))))

(assert (forall ((l0 Int) (l Int) (h Int) (i Int) (p Int) (b Int) (v Int) (f Int) (k Int) (ak Int) (k2 Int) (ak2 Int))
  (=> (and (test l0 l h i p b v f k ak k2 ak2)
           (< v b))
      (loop l0 l h (+ i 1) i v f k ak k2 ak2))))

(assert (forall ((l0 Int) (l Int) (h Int) (i Int) (p Int) (b Int) (v Int) (f Int) (k Int) (ak Int) (k2 Int) (ak2 Int))
  (=> (and (test l0 l h i p b v f k ak k2 ak2)
           (>= v b))
      (loop l0 l h (+ i 1) p b f k ak k2 ak2))))

(assert (forall ((l0 Int) (l Int) (h Int) (i Int) (p Int) (b Int) (f Int) (k Int) (ak Int) (k2 Int) (ak2 Int))
  (=> (and (loop l0 l h i p b f k ak k2 ak2)
           (>= i h))
      (write1 l0 l h p b f k ak k2 ak2))))

; The Swap, 1st write
(assert (forall ((l0 Int) (l Int) (h Int) (p Int) (b Int) (f Int) (k Int) (ak Int) (k2 Int) (ak2 Int))
  (=> (and (distinct l k) (distinct l k2)
           (write1 l0 l h p b f k ak k2 ak2))
      (write2 l0 l h p b f k ak k2 ak2))))

(assert (forall ((l0 Int) (l Int) (h Int) (p Int) (b Int) (f Int) (ak Int) (k2 Int) (ak2 Int))
  (=> (and (distinct l k2)
           (write1 l0 l h p b f l ak k2 ak2))
      (write2 l0 l h p b f l b k2 ak2))))

(assert (forall ((l0 Int) (l Int) (h Int) (p Int) (b Int) (f Int) (k Int) (ak Int) (ak2 Int))
  (=> (and (distinct l k)
           (write1 l0 l h p b f k ak l ak2))
      (write2 l0 l h p b f k ak l b))))

(assert (forall ((l0 Int) (l Int) (h Int) (p Int) (b Int) (f Int) (ak Int))
  (=> (write1 l0 l h p b f l ak l ak)
      (write2 l0 l h p b f l b l b))))

; The Swap, 2nd write
(assert (forall ((l0 Int) (l Int) (h Int) (p Int) (b Int) (f Int) (k Int) (ak Int) (k2 Int) (ak2 Int))
  (=> (and (distinct p k) (distinct p k2)
           (write2 l0 l h p b f k ak k2 ak2))
      (endswap l0 l h p b f k ak k2 ak2))))

(assert (forall ((l0 Int) (l Int) (h Int) (p Int) (b Int) (f Int) (ak Int) (k2 Int) (ak2 Int))
  (=> (and (distinct p k2)
           (write2 l0 l h p b f p ak k2 ak2))
      (endswap l0 l h p b f p f k2 ak2))))

(assert (forall ((l0 Int) (l Int) (h Int) (p Int) (b Int) (f Int) (k Int) (ak Int) (ak2 Int))
  (=> (and (distinct p k)
           (write2 l0 l h p b f k ak p ak2))
      (endswap l0 l h p b f k ak p f))))

(assert (forall ((l0 Int) (l Int) (h Int) (p Int) (b Int) (f Int) (ak Int))
  (=> (write2 l0 l h p b f p ak p ak)
      (endswap l0 l h p b f p f p f))))

; incr and outerloop
(assert (forall ((l0 Int) (l Int) (h Int) (p Int) (b Int) (f Int) (k Int) (ak Int) (k2 Int) (ak2 Int))
  (=> (endswap l0 l h p b f k ak k2 ak2)
      (incr l0 l h k ak k2 ak2))))

(assert (forall ((l0 Int) (l Int) (h Int) (k Int) (ak Int) (k2 Int) (ak2 Int))
  (=> (incr l0 l h k ak k2 ak2)
      (outerloop l0 (+ l 1) h k ak k2 ak2))))

; Initialization
(assert (forall ((l Int) (h Int) (k Int) (ak Int) (k2 Int) (ak2 Int))
  (=> (< k k2) (init l h k ak k2 ak2))))
(assert (forall ((l Int) (h Int) (k Int) (ak Int))
  (init l h k ak k ak)))

; Various invariants

; Removing this one yields 18.2s with Z3/PDR b1d649fe1c842208f701c6b1d1dfa5d17e8dc679
;(assert (forall ((l0 Int) (l Int) (h Int) (p Int) (b Int) (f Int) (k Int) (ak I;nt) (k2 Int) (ak2 Int))
;  (=> (write1 l0 l h p b f k ak k2 ak2)
;      (>= p l))))

; Removing the first two ones yields 10.5s with Z3/PDR b1d649fe1c842208f701c6b1d1dfa5d17e8dc679
;(assert (forall ((l0 Int) (l Int) (h Int) (p Int) (b Int) (f Int) (k Int) (ak Int) (k2 Int) (ak2 Int))
;  (=> (write1 l0 l h p b f k ak k2 ak2)
;      (< p h))))

; Removing the first three ones yields 11.5s with Z3/PDR b1d649fe1c842208f701c6b1d1dfa5d17e8dc679
;(assert (forall ((l0 Int) (l Int) (h Int) (p Int) (b Int) (f Int) (ap Int))
;  (=> (write1 l0 l h p b f p ap p ap)
;      (= b ap))))

; Removing the first four ones yields 17.1s with Z3/PDR b1d649fe1c842208f701c6b1d1dfa5d17e8dc679
;(assert (forall ((l0 Int) (l Int) (h Int) (p Int) (b Int) (f Int) (al Int))
;  (=> (write1 l0 l h p b f l al l al)
;      (= f al))))

; Removing the first five ones yields 6.2s with Z3/PDR b1d649fe1c842208f701c6b1d1dfa5d17e8dc679
;(assert (forall ((l0 Int) (l Int) (h Int) (p Int) (b Int) (f Int) (k Int) (ak Int))
;  (=> (and (>= k l)
;           (< k h)
;           (write1 l0 l h p b f k ak k ak))
;      (<= b ak))))

; Removing the first six ones yields 6.7s with Z3/PDR b1d649fe1c842208f701c6b1d1dfa5d17e8dc679
;(assert (forall ((l0 Int) (l Int) (h Int) (p Int) (b Int) (f Int) (k Int) (ak Int) (al Int))
;  (=> (and (>= k l)
;           (< k h)
;           (endswap l0 l h p b f l al k ak))
;      (<= al ak))))

; Removing the first seven ones yields 4.6s with Z3/PDR b1d649fe1c842208f701c6b1d1dfa5d17e8dc679
; 8.5s in Z3/PDR 2015-06-01_168ea2e948bf6d946e163a5a1af0ddf084552cd6
; 5.5s in Z3/Spacer 2014-08-06_b1d649fe1c842208f701c6b1d1dfa5d17e8dc679
;(assert (forall ((l0 Int) (l Int) (h Int) (k Int) (ak Int) (al Int))
;  (=> (and (>= k l)
;           (< k h)
;           (incr l0 l h l al k ak))
;      (<= al ak))))

; This hint seems REALLY NECESSARY: with it
; Z/Spacer 2014-08-06_b1d649fe1c842208f701c6b1d1dfa5d17e8dc679 takes 2.4s
; without it, no convergence after 9min

; Removing this extra condition yields UNSAT!? with Z3/PDR b1d649fe1c842208f701c6b1d1dfa5d17e8dc679
;(assert (forall ((l0 Int) (l Int) (h Int) (k Int) (ak Int) (k2 Int) (ak2 Int))
;  (=> (and (>= k l0) (< k l) (>= k2 k) (< k2 h)
;           (outerloop l0 l h k ak k2 ak2))
;      (<= ak ak2))))

; This one greatly helps proving the sortedness property with k<=k2
; instead of k<k2 !
;(assert (forall ((l0 Int) (l Int) (h Int) (k Int) (ak Int) (k Int) (ak2 Int))
;  (=> (and (>= k l0) (< k h) (outerloop l0 l h k ak k ak2))
;      (= ak ak2))))

; Final: sortedness (time 2s -> 14s if using k <= k2)
(assert (forall ((l0 Int) (h Int) (k Int) (ak Int) (k2 Int) (ak2 Int))
  (=> (and (>= k l0) (< k k2) (< k2 h)
           (exit l0 h k ak k2 ak2))
      (<= ak ak2))))

; with this condition, should be UNSAT
;(assert (forall ((l0 Int) (h Int) (k Int) (ak Int) (k2 Int) (ak2 Int))
;  (not (exit l0 h l0 ak (+ l0 1) ak2))))

(check-sat)
\end{lstlisting}

\begin{lstlisting}[language=SMT,caption={Selection sort: permutation},label=Horn:selection_sort_multiset]
; Status: SAT in 9s by Spacer 2015-07-01_0dfed9c654a6f715b8660ff1d27de74fbfc41916

(set-logic HORN)

(declare-fun init (Int Int Int Int Int) Bool) ; l0 h k a[k] z #a0(z)
(declare-fun outerloop (Int Int Int Int Int Int Int Int) Bool) ; l0 l h k a[k] z #a(z) #a0(z)
(declare-fun exit (Int Int Int Int Int Int Int) Bool) ; l0 h k a[k] z #a(z) #a0(z)

(declare-fun read1 (Int Int Int Int Int Int Int Int) Bool) ; l0 l h k a[k] z #a(z) #a0(z)
(declare-fun loop (Int Int Int Int Int Int Int Int Int Int Int Int) Bool) ; l0 l h i p b f k a[k] z #a(z) #a0(z)
(declare-fun read2 (Int Int Int Int Int Int Int Int Int Int Int Int) Bool) ; l0 l h i p b f k a[k] z #a(z) #a0(z)
(declare-fun test (Int Int Int Int Int Int Int Int Int Int Int Int Int) Bool) ; l0 l h i p b v f k a[k] z #a(z) #a0(z)

(declare-fun write1 (Int Int Int Int Int Int Int Int Int Int Int) Bool) ; l0 l h p b f k a[k] z #a(z) #a0(z)
(declare-fun write1a (Int Int Int Int Int Int Int Int Int Int Int) Bool) ; l0 l h p b f k a[k] z #a(z) #a0(z)
(declare-fun write1b (Int Int Int Int Int Int Int Int Int Int Int) Bool) ; l0 l h p b f k a[k] z #a(z) #a0(z)

(declare-fun write2 (Int Int Int Int Int Int Int Int Int Int Int) Bool) ; l0 l h p b f k a[k] z #a(z) #a0(z)
(declare-fun write2a (Int Int Int Int Int Int Int Int Int Int Int) Bool) ; l0 l h p b f k a[k] z #a(z) #a0(z)
(declare-fun write2b (Int Int Int Int Int Int Int Int Int Int Int) Bool) ; l0 l h p b f k a[k] z #a(z) #a0(z)

(declare-fun endswap (Int Int Int Int Int Int Int Int Int Int Int) Bool) ; l0 l h p b f k a[k] z #a(z) #a0(z)
(declare-fun incr (Int Int Int Int Int Int Int Int) Bool) ; l0 l h k a[k] z #a(z) #a0(z)

(assert (forall ((l0 Int) (h Int) (k Int) (ak Int) (z Int) (a0z Int))
  (init l0 h k z a0z)))

(assert (forall ((l0 Int) (h Int) (k Int) (ak Int) (z Int) (a0z Int))
   (=> (and (init l0 h k z a0z) (<= l0 h))
       (outerloop l0 l0 h k ak z a0z a0z))))

(assert (forall ((l0 Int) (l Int) (h Int) (k Int) (ak Int) (z Int) (az Int) (a0z Int))
   (=> (and (outerloop l0 l h k ak z az a0z) (< l (- h 1)))
       (read1 l0 l h k ak z az a0z))))

(assert (forall ((l0 Int) (l Int) (h Int) (k Int) (ak Int) (z Int) (az Int) (a0z Int))
   (=> (and (outerloop l0 l h k ak z az a0z) (>= l (- h 1)))
       (exit l0 h k ak z az a0z))))

(assert (forall ((l0 Int) (l Int) (h Int) (b Int) (k Int) (ak Int) (z Int) (az Int) (a0z Int))
   (=> (and (read1 l0 l h k ak z az a0z)
            (read1 l0 l h l b z az a0z)
            (distinct k l))
       (loop l0 l h (+ l 1) l b b k ak z az a0z))))

(assert (forall ((l0 Int) (l Int) (h Int) (b Int) (z Int) (az Int) (a0z Int))
   (=> (read1 l0 l h l b z az a0z)
       (loop l0 l h (+ l 1) l b b l b z az a0z))))

(assert (forall ((l0 Int) (l Int) (h Int) (i Int) (p Int) (b Int) (f Int) (k Int) (ak Int) (z Int) (az Int) (a0z Int))
   (=> (and (loop l0 l h i p b f k ak z az a0z) (< i h))
       (read2 l0 l h i p b f k ak z az a0z))))

(assert (forall ((l0 Int) (l Int) (h Int) (i Int) (p Int) (b Int) (f Int) (k Int) (ak Int) (z Int) (az Int) (a0z Int) (v Int))
   (=> (and (read2 l0 l h i p b f k ak z az a0z)
            (read2 l0 l h i p b f i v z az a0z)
            (distinct k i))
       (test l0 l h i p b v f k ak z az a0z))))

(assert (forall ((l0 Int) (l Int) (h Int) (i Int) (p Int) (b Int) (v Int) (f Int) (z Int) (az Int) (a0z Int))
  (=> (read2 l0 l h i p b f i v z az a0z)
      (test l0 l h i p b v f i v z az a0z))))

(assert (forall ((l0 Int) (l Int) (h Int) (i Int) (p Int) (b Int) (v Int) (f Int) (k Int) (ak Int) (z Int) (az Int) (a0z Int))
  (=> (and (test l0 l h i p b v f k ak z az a0z)
           (< v b))
      (loop l0 l h (+ i 1) i v f k ak z az a0z))))

(assert (forall ((l0 Int) (l Int) (h Int) (i Int) (p Int) (b Int) (v Int) (f Int) (k Int) (ak Int) (z Int) (az Int) (a0z Int))
  (=> (and (test l0 l h i p b v f k ak z az a0z)
           (>= v b))
      (loop l0 l h (+ i 1) p b f k ak z az a0z))))

(assert (forall ((l0 Int) (l Int) (h Int) (i Int) (p Int) (b Int) (f Int) (k Int) (ak Int) (z Int) (az Int) (a0z Int))
  (=> (and (loop l0 l h i p b f k ak z az a0z)
           (>= i h))
      (write1 l0 l h p b f k ak z az a0z))))

; The Swap, 1st write
(assert (forall ((l0 Int) (l Int) (h Int) (p Int) (b Int) (f Int) (k Int) (al Int) (ak Int) (z Int) (az Int) (a0z Int))
  (=> (and (write1 l0 l h p b f k ak z az a0z)
           (write1 l0 l h p b f l al z az a0z)
           (distinct al z))
      (write1a l0 l h p b f k ak z az a0z))))

(assert (forall ((l0 Int) (l Int) (h Int) (p Int) (b Int) (f Int) (k Int) (al Int) (ak Int) (z Int) (az Int) (a0z Int))
  (=> (and (write1 l0 l h p b f k ak al az a0z)
           (write1 l0 l h p b f l al al az a0z))
      (write1a l0 l h p b f k ak al (- az 1) a0z))))

(assert (forall ((l0 Int) (l Int) (h Int) (p Int) (b Int) (f Int) (k Int) (al Int) (ak Int) (z Int) (az Int) (a0z Int))
  (=> (and (write1a l0 l h p b f k ak z az a0z)
           (write1a l0 l h p b f l al z az a0z)
           (distinct b z))
      (write1b l0 l h p b f k ak z az a0z))))

(assert (forall ((l0 Int) (l Int) (h Int) (p Int) (b Int) (f Int) (k Int) (al Int) (ak Int) (z Int) (az Int) (a0z Int))
  (=> (and (write1a l0 l h p b f k ak b az a0z)
           (write1a l0 l h p b f l al b az a0z))
      (write1b l0 l h p b f k ak b (+ az 1) a0z))))

(assert (forall ((l0 Int) (l Int) (h Int) (p Int) (b Int) (f Int) (k Int) (ak Int) (z Int) (az Int) (a0z Int))
  (=> (and (distinct l k)
           (write1b l0 l h p b f k ak z az a0z))
      (write2 l0 l h p b f k ak z az a0z))))

(assert (forall ((l0 Int) (l Int) (h Int) (p Int) (b Int) (f Int) (ak Int) (z Int) (az Int) (a0z Int))
  (=> (write1b l0 l h p b f l ak z az a0z)
      (write2 l0 l h p b f l b z az a0z))))

; The Swap, 2nd write
(assert (forall ((l0 Int) (l Int) (h Int) (p Int) (b Int) (f Int) (k Int) (ap Int) (ak Int) (z Int) (az Int) (a0z Int))
  (=> (and (write2 l0 l h p b f k ak z az a0z)
           (write2 l0 l h p b f p ap z az a0z)
           (distinct ap z))
      (write2a l0 l h p b f k ak z az a0z))))

(assert (forall ((l0 Int) (l Int) (h Int) (p Int) (b Int) (f Int) (k Int) (ap Int) (ak Int) (z Int) (az Int) (a0z Int))
  (=> (and (write2 l0 l h p b f k ak ap az a0z)
           (write2 l0 l h p b f p ap ap az a0z))
      (write2a l0 l h p b f k ak ap (- az 1) a0z))))

(assert (forall ((l0 Int) (l Int) (h Int) (p Int) (b Int) (f Int) (k Int) (ap Int) (ak Int) (z Int) (az Int) (a0z Int))
  (=> (and (write2a l0 l h p b f k ak z az a0z)
           (write2a l0 l h p b f p ap z az a0z)
           (distinct f z))
      (write2b l0 l h p b f k ak z az a0z))))

(assert (forall ((l0 Int) (l Int) (h Int) (p Int) (b Int) (f Int) (k Int) (ap Int) (ak Int) (z Int) (az Int) (a0z Int))
  (=> (and (write2a l0 l h p b f k ak f az a0z)
           (write2a l0 l h p b f p ap f az a0z))
      (write2b l0 l h p b f k ak f (+ az 1) a0z))))

(assert (forall ((l0 Int) (l Int) (h Int) (p Int) (b Int) (f Int) (k Int) (ak Int) (z Int) (az Int) (a0z Int))
  (=> (and (distinct p k)
           (write2b l0 l h p b f k ak z az a0z))
      (endswap l0 l h p b f k ak z az a0z))))

(assert (forall ((l0 Int) (l Int) (h Int) (p Int) (b Int) (f Int) (ak Int) (z Int) (az Int) (a0z Int))
  (=> (write2b l0 l h p b f p ak z az a0z)
      (endswap l0 l h p b f p f z az a0z))))

; incr and outerloop
(assert (forall ((l0 Int) (l Int) (h Int) (p Int) (b Int) (f Int) (k Int) (ak Int) (z Int) (az Int) (a0z Int))
  (=> (endswap l0 l h p b f k ak z az a0z)
      (incr l0 l h k ak z az a0z))))

(assert (forall ((l0 Int) (l Int) (h Int) (k Int) (ak Int) (z Int) (az Int) (a0z Int))
  (=> (incr l0 l h k ak z az a0z)
      (outerloop l0 (+ l 1) h k ak z az a0z))))

; Various invariants (hints)

;; (assert (forall ((l0 Int) (l Int) (h Int) (p Int) (b Int) (f Int) (k Int) (ak Int) (z Int) (az Int) (a0z Int))
;;   (=> (write1 l0 l h p b f k ak z az a0z)
;;       (>= p l))))

;; (assert (forall ((l0 Int) (l Int) (h Int) (p Int) (b Int) (f Int) (k Int) (ak Int) (z Int) (az Int) (a0z Int))
;;   (=> (write1 l0 l h p b f k ak z az a0z)
;;       (< p h))))

;; (assert (forall ((l0 Int) (l Int) (h Int) (p Int) (b Int) (f Int) (ap Int) (z Int) (az Int) (a0z Int))
;;   (=> (write1 l0 l h p b f p ap z az a0z)
;;       (= b ap))))

;; (assert (forall ((l0 Int) (l Int) (h Int) (p Int) (b Int) (f Int) (al Int) (z Int) (az Int) (a0z Int))
;;   (=> (write1 l0 l h p b f l al z az a0z)
;;       (= f al))))

;; (assert (forall ((l0 Int) (l Int) (h Int) (p Int) (b Int) (f Int) (k Int) (ak Int) (z Int) (az Int) (a0z Int))
;;   (=> (and (>= k l)
;;            (< k h)
;;            (write1 l0 l h p b f k ak z az a0z))
;;       (<= b ak))))

;; (assert (forall ((l0 Int) (l Int) (h Int) (p Int) (b Int) (f Int) (k Int) (ak Int) (z Int) (az Int) (a0z Int))
;;   (=> (write1 l0 l h p b f k ak z az a0z)
;;       (= az a0z))))

;; (assert (forall ((l0 Int) (l Int) (h Int) (p Int) (b Int) (f Int) (k Int) (ak Int) (z Int) (az Int) (a0z Int))
;;   (=> (endswap l0 l h p b f k ak z az a0z)
;;       (= az a0z))))

; Final property
(assert (forall ((l0 Int) (h Int) (k Int) (ak Int) (z Int) (az Int) (a0z Int))
  (=> (exit l0 h k ak z az a0z) (= az a0z))))

(check-sat)
\end{lstlisting}

\end{document}